\documentclass[a4paper]{amsart}

\usepackage{amsmath,amssymb,mathtools,mathrsfs}
\usepackage[UKenglish]{babel}
\usepackage[colorlinks=true,urlcolor=blue,citecolor=red,linkcolor=blue,linktocpage,pdfpagelabels,bookmarksnumbered,bookmarksopen]{hyperref}
\usepackage{enumerate,paralist}
\usepackage[dvips]{graphicx}
	\graphicspath{{./}{figures/}}
\allowdisplaybreaks

\DeclarePairedDelimiter{\abs}{\lvert}{\rvert}
\DeclarePairedDelimiter{\norm}{\lVert}{\rVert}

\makeatletter
\let\oldabs\abs
\def\abs{\@ifstar{\oldabs}{\oldabs*}}
\let\oldnorm\norm
\def\norm{\@ifstar{\oldnorm}{\oldnorm*}}
\makeatother

\graphicspath{{./}{figure/}}

\newcommand{\cI}{\mathscr{I}}
\newcommand{\Lip}{\operatorname{Lip}}
\newcommand{\cM}[1]{\mathcal{M}_+^{#1}(I^2)}
\newcommand{\cP}{\mathcal{P}}
\newcommand{\sP}{\mathscr{P}}
\newcommand{\Prob}{\operatorname{Prob}}
\newcommand{\R}{\mathbb{R}}
\newcommand{\cS}{\mathscr{S}}
\newcommand{\w}{\mathbf{w}}
\newcommand{\wass}[2]{W_1\left(#1,\,#2\right)}

\theoremstyle{remark}\newtheorem{remark}{Remark}[section]
\theoremstyle{plain}\newtheorem{theorem}[remark]{Theorem}
\theoremstyle{definition}
	\newtheorem{assumption}[remark]{Assumption}
	\newtheorem{definition}[remark]{Definition}

\title[Proposal of a risk model for vehicular traffic]{Proposal of a risk model for vehicular traffic: A Boltzmann-type kinetic approach}
\author{	Paolo Freguglia}
\address{Department of Engineering, Information Sciences, and Mathematics, University of L'Aquila, Via Vetoio (Coppito 1), 67100 Coppito AQ, Italy}
\email{pfreguglia@gmail.com}
\author{Andrea Tosin}
\address{Department of Mathematical Sciences ``G. L. Lagrange'', Politecnico di Torino, Corso Duca degli Abruzzi 24, 10129 Turin, Italy}
\email{andrea.tosin@polito.it}

\begin{document}



\begin{abstract}
This paper deals with a Boltzmann-type kinetic model describing the interplay between vehicle dynamics and safety aspects in vehicular traffic. Sticking to the idea that the macroscopic characteristics of traffic flow, including the distribution of the driving risk along a road, are ultimately generated by one-to-one interactions among drivers, the model links the personal (i.e., individual) risk to the changes of speeds of single vehicles and implements a probabilistic description of such microscopic interactions in a Boltzmann-type collisional operator. By means of suitable statistical moments of the kinetic distribution function, it is finally possible to recover macroscopic relationships between the average risk and the road congestion, which show an interesting and reasonable correlation with the well-known free and congested phases of the flow of vehicles.
\end{abstract}

\maketitle

\section{Introduction}
Road safety is a major issue in modern societies, especially in view of the constantly increasing motorisation levels across several EU and non-EU countries. Although recent studies suggest that this fact is actually correlated with a general decreasing trend of fatality rates, see e.g.,~\cite{DaCoTA2012,yannis2011TRB}, the problem of assessing quantitatively the risk in vehicular traffic, and of envisaging suitable countermeasures, remains of paramount importance.

So far, road safety has been studied mainly by means of statistical models aimed at fitting the probability distribution of the fatality rates over time~\cite{oppe1989AAP} or at forecasting road accidents using time series~\cite{abdel-aty2000AAP,miaou1993AAP}. Efforts have also been made towards the construction of safety indicators, which should allow one to classify the safety performances of different roads and to compare, on such a basis, different countries~\cite{hermans2009AAP,hermans2008AAP}. However, there is in general no agreement on which procedure, among several possible ones, is the most suited to construct a reliable indicator and, as a matter of fact, the position of a given country in the ranking turns out to be very sensitive to the indicator used. Despite this, a synthetic analysis is ultimately necessary: a mere comparison of the crash data of different countries may be misleading, therefore a more abstract and comprehensive concept of \emph{risk} has to be formulated~\cite{shen2012AAP}.

A recent report on road safety in New Zealand~\cite{KiwiRAP2012} introduces the following definitions of two types of risk:
\begin{description}
\item[Collective risk] is a measure of the total number of fatal and serious injury crashes per kilometre over a section of road, cf.~\cite[p. 13]{KiwiRAP2012};
\item[Personal risk] is a measure of the danger to each individual using the state highway being assessed, cf.~\cite[p. 14]{KiwiRAP2012}.
\end{description}
While substantially qualitative and empirical, these definitions raise nevertheless an important conceptual point, namely the fact that the risk is intrinsically \emph{multiscale}. Each driver (\emph{microscopic} scale) bears a certain personal level of danger, namely of potential risk, which, combined with the levels of danger of all other drivers, forms an emergent risk for the indistinct ensemble of road users (\emph{macroscopic} scale). Hence the large-scale tangible manifestations of the road risk originate from small-scale, often unobservable, causes. Such an argument is further supported by some psychological theories of risk perception, among which probably the most popular one in the context of vehicular traffic is the so-called \emph{risk homeostasis theory}. According to this theory, each driver possesses a certain target level of personal risk, which s/he feels comfortable with; then, at every time s/he compares their perceived risk with such a target level, adjusting their behaviour so as to reduce the gap between the two~\cite{wilde1998IP}. Actually, the risk homeostasis theory is not widely accepted, some studies rejecting it on the basis of experimental evidences, see e.g.,~\cite{evans1986RA}. The main criticism is, in essence, that the aforesaid risk regulatory mechanism of the drivers (acting similarly to the thermal homeostatic system in warm-blooded animals, whence the name of the theory) is too elementary compared to the much richer variety of possible responses, to such an extent that some paradoxical consequences are produced. For instance, the number of traffic accidents per unit time would tend to be constant independently of possible safety countermeasures, because so tends to be the personal risk per unit time. Whether one accepts or not this theory, there is a common agreement on the fact that the background of all observable manifestations of the road risk is the individual behaviour of the drivers. In this respect, conceiving a mathematical model able to explore the link between small and large scale effects acquires both a theoretical and a practical interest. In fact, if on the one hand data collection is a useful practice in order to grasp the essential trends of the considered phenomenon, on the other hand the interpretation of the data themselves, with possibly the goal of making simulations and predictions, cannot rely simply on empirical observation.

The mathematical literature offers nowadays a large variety of traffic models at all observation and representation scales, from the microscopic and kinetic to the macroscopic one, see e.g.,~\cite{piccoli2009ENCYCLOPEDIA} and references therein for a critical survey. Nevertheless, there is a substantial lack of models dedicated to the joint simulation of traffic flow and safety issues. In~\cite{moutari2013IMAJAM} the authors propose a model, which is investigated analytically in~\cite{herty2011ZAMM} and then further improved in~\cite{moutari2014CMS}, for the simulation of car accidents. The model is a macroscopic one based on the coupling of two second order traffic models, which are instantaneously defined on two disjoint adjacent portions of the road and which feature different traffic pressure laws accounting for more and less careful drivers. Car collisions are understood as the intersection of the trajectories of two vehicles driven by either type of driver. In particular, analytical conditions are provided, under which a collision occurs depending on the initial space and speed distributions of the vehicles.

In this paper, instead of modelling physical collisions among cars, we are more interested in recovering the point of view based on the concept of risk discussed at the beginning. Sticking to the idea that observable traffic trends are ultimately determined by the individual behaviour of drivers, we adopt a Boltzmann-type kinetic approach focusing on binary interactions among drivers, which are responsible for both speed changes (through instantaneous acceleration, braking, overtaking) and, consequently, also for changes in the individual levels of risk. In practice, as usual in the kinetic theory approach, we consider the time-evolution of the statistical distribution of the microscopic states of the vehicles, taking into account that such states include also the \emph{personal risk} of the drivers. Then, by extracting suitable mean quantities at equilibrium from the kinetic distribution function, we obtain some information about the macroscopic traffic trends, including the \emph{average risk} and the \emph{probability of accident} as functions of the road congestion. To some extent, these can be regarded as measures of a \emph{potential collective risk} useful to both road users and traffic governance authorities.

In more detail, the paper is organised as follows. In Section~\ref{sec:model} we present the Boltzmann-type kinetic model and specialise it to the case of a \emph{quantised} space of microscopic states, given that the vehicle speed and the personal risk can be conveniently understood as discrete variables organised in levels. In Section~\ref{sec:interactions} we detail the modelling of the microscopic interactions among the vehicles, regarding them as stochastic jump processes on the discrete state space. Indeed, the aforementioned variety of human responses suggests that a probabilistic approach is more appropriate at this stage. In Section~\ref{sec:simulations} we perform a computational analysis of the model, which leads us to define the \emph{risk diagram} of traffic parallelly to the more celebrated fundamental and speed diagrams. By means of such a diagram we display the average risk as a function of the vehicle density and we take inspiration for proposing the definition of a \emph{safety criterion} which discriminates between \emph{safety} and \emph{risk regimes} depending on the local traffic congestion. Interestingly enough, such regimes turn out to be correlated with the well-known phase transition between free and congested flow regimes also reproduced by our model. In Section~\ref{sec:conclusions} we draw some conclusions and briefly sketch research perspectives regarding the application of ideas similar to those developed in this paper to other systems of interacting particles prone to safety issues, for instance crowds. Finally, in Appendix~\ref{app:basictheo} we develop a basic well-posedness and asymptotic theory of our kinetic model in measure spaces (Wasserstein spaces), so as to ground the contents of the paper on solid mathematical bases.

\section{Boltzmann-type kinetic model with stochastic interactions}
\label{sec:model}
In this section we introduce a model based on Boltzmann-type kinetic equations, in which short-range interactions among drivers are modelled as stochastic transitions of microscopic state. This allows us to introduce the randomness of the human behaviour in the microscopic dynamics ruling the individual response to local traffic and safety conditions.

In particular, we consider the following (dimensionless) microscopic states of the drivers: the \emph{speed} $v\in [0,\,1]\subset\R$ and the \emph{personal risk} $u\in [0,\,1]\subset\R$, with $u=0$, $u=1$ standing for the lowest and highest risk, respectively. The kinetic (statistical) representation of the system is given by the \emph{(one particle) distribution function} over the microscopic states, say $f=f(t,\,v,\,u)$, $t\geq 0$ being the time variable, such that $f(t,\,v,\,u)\,dv\,du$ is the fraction of vehicles which at time $t$ travel at a speed comprised between $v$ and $v+dv$ with a personal risk between $u$ and $u+du$. Alternatively, if the distribution function $f$ is thought of as normalised with respect to the total number of vehicles, $f(t,\,v,\,u)\,dv\,du$ can be understood as the probability that a representative vehicle of the system possesses a microscopic state in $[v,\,v+dv]\times [u,\,u+du]$ at time $t$.

\begin{remark} \label{rem:u}
The numerical values of $u$ introduced above are purely conventional: they are used to mathematise the concept of personal risk but neither refer nor imply actual physical ranges. Hence they serve mostly \emph{interpretative} than strictly quantitative purposes: for instance, the definition and identification of the macroscopic risk and safety regimes of traffic, cf. Section~\ref{sec:riskdiag}. The quantitative information on these regimes will be linked to a more standard and well-defined physical quantity, such as the vehicle density.
\end{remark}

\begin{remark}
The above statistical representation does not include the space coordinate among the variables which define the microscopic state of the vehicles. This is because we are considering a simplified setting, in which the distribution of the vehicles along the road is supposed to be \emph{homogeneous}. While being a rough physical approximation, this assumption nonetheless allows us to focus on the interaction dynamics among vehicles, which feature mainly speed variations, here further linked to variations of the personal risk. Hence the two microscopic variables introduced above are actually the most relevant ones for constructing a minimal mathematical model which describes traffic dynamics as a result of concurrent mechanical and behavioural effects.
\end{remark}

\subsection{Use of the distribution function for computing observable quantities}
If the distribution function $f$ is known, several statistics over the microscopic states of the vehicles can be computed. Such statistics provide macroscopic observable quantities related to both the traffic conditions and the risk along the road.

We recall in particular some of them, which will be useful in the sequel:
\begin{itemize}
\item The \emph{density} of vehicles at time $t$, denoted $\rho=\rho(t)$, which is defined as the zeroth order moment of $f$ with respect to both $v$ and $u$:
$$ \rho(t):=\int_0^1\int_0^1 f(t,\,v,\,u)\,dv\,du. $$
Throughout the paper we will assume $0\leq\rho\leq 1$, where $\rho=1$ represents the maximum (dimensionless) density of vehicles that can be locally accommodated on the road.
\item The \emph{average flux} of vehicles at time $t$, denoted $q=q(t)$, which is defined as the first order moment of the speed distribution, the latter being the marginal of $f$ with respect to $u$. Hence:
\begin{equation}
	q(t)=\int_0^1 v\left(\int_0^1 f(t,\,v,\,u)\,du\right)\,dv.
	\label{eq:ave_flux}
\end{equation}
\item The \emph{mean speed} of vehicles at time $t$, denoted $V=V(t)$, which is defined from the usual relationship between $\rho$ and $q$, i.e., $q=\rho V$, whence
\begin{equation}
	V(t)=\frac{q(t)}{\rho(t)}.
	\label{eq:mean_speed}
\end{equation}
\item The \emph{statistical distribution of the risk}, say $\varphi=\varphi(t,\,u)$, namely the marginal of $f$ with respect to $v$:
\begin{align}
	\begin{aligned}[t]
		\varphi(t,\,u):=\int_0^1f(t,\,v,\,u)\,dv,
	\end{aligned}
	\label{eq:risk_distr}
\end{align}
which is such that $\varphi(t,\,u)\,du$ is the number of vehicles which, at time $t$, bear a personal risk comprised between $u$ and $u+du$ (regardless of their speed). Using $\varphi$ we obtain the \emph{average risk}, denoted $U(t)$, along the road at time $t$ as:
\begin{align}
	\begin{aligned}[t]
		U(t):=\frac{1}{\rho(t)}\int_0^1u\varphi(t,\,u)\,du=\frac{1}{\rho(t)}\int_0^1 u\left(\int_0^1 f(t,\,v,\,u)\,dv\right)\,du.
	\end{aligned}
	\label{eq:ave_risk}
\end{align}
Notice that $\int_0^1\varphi(t,\,u)\,du=\rho(t)$, which explains the coefficient $\frac{1}{\rho(t)}$ in~\eqref{eq:ave_risk}.
\end{itemize}

\subsection{Evolution equation for the distribution function}
A mathematical mo\-del consists in an evolution equation for the distribution function $f$, derived consistently with the principles of the kinetic theory of vehicular traffic.

In our spatially homogeneous setting, the time variation of the number of vehicles with speed $v$ and personal risk $u$ is only due to short-range interactions, which cause acceleration and braking. Since the latter depend ultimately on people's driving style, they cannot be modelled by appealing straightforwardly to standard mechanical principles. Hence, in the following, speed and risk transitions will be regarded as correlated stochastic processes. This way both the subjectivity of the human behaviour and the interplay between mechanical and behavioural effects will be taken into account.

In formulas we write
$$ \partial_tf=Q^{+}(f,\,f)-fQ^{-}(f), $$
where:
\begin{itemize}
\item $Q^{+}(f,\,f)$ is a bilinear \emph{gain operator}, which counts the average number of interactions per unit time originating new vehicles with post-interaction state $(v,\,u)$;
\item $Q^{-}(f)$ is a linear \emph{loss operator}, which counts the average number of interactions per unit time causing vehicles with pre-interaction state $(v,\,u)$ to change either the speed or the personal risk.
\end{itemize}
Let us introduce the following compact notations: $\w:=(v,\,u)$, $I=[0,\,1]\subset\R$. Then the expression of the gain operator is as follows (see e.g.,~\cite{tosin2009AML}):
\begin{equation}
	Q^{+}(f,\,f)(t,\,\w)=\iint_{I^4}\eta(\w_\ast,\,\w^\ast)\cP(\w_\ast\to\w\,\vert\,\w^\ast,\,\rho)f(t,\,\w_\ast)f(t,\,\w^\ast)\,d\w_\ast\,d\w^\ast
	\label{eq:gain}
\end{equation}
where $\w_\ast,\,\w^\ast$ are the pre-interaction states of the two vehicles which interact and:
\begin{itemize}
\item $\eta(\w_\ast,\,\w^\ast)>0$ is the \emph{interaction rate}, i.e., the frequency of interaction between a vehicle with microscopic state $\w_\ast$ and one with microscopic state $\w^\ast$;
\item $\cP(\w_\ast\to\w\,\vert\,\w^\ast,\,\rho)\geq 0$ is the \emph{transition probability distribution}. More precisely, $\cP(\w_\ast\to\w\,\vert\,\w^\ast,\,\rho)\,d\w$ is the probability that the vehicle with microscopic state $\w_\ast$ switches to a microscopic state contained in the elementary volume $d\w$ of $I^2$ centred in $\w$ because of an interaction with the vehicle with microscopic state $\w^\ast$. Conditioning by the density $\rho$ indicates that, as we will see later (cf. Section~\ref{sec:interactions}), binary interactions are influenced by the local macroscopic state of the traffic.

For fixed pre-interaction states the following property holds:
\begin{equation}
	\int_{I^2}\cP(\w_\ast\to\w\,\vert\,\w^\ast,\,\rho)\,d\w=1 \qquad \forall\,\w_\ast,\,\w^\ast\in I^2,\ \forall\,\rho\in [0,\,1].
	\label{eq:tabgames}
\end{equation}
\end{itemize}

Likewise, the expression of the loss operator is as follows (see again~\cite{tosin2009AML}):
$$ Q^{-}(f)(t,\,\w)=\int_{I^2}\eta(\w,\,\w^\ast)f(t,\,\w^\ast)\,d\w^\ast. $$
On the whole, the loss term $fQ^{-}(f)$ can be derived from the gain term~\eqref{eq:gain} by assuming that the first vehicle already holds the state $\w$ and counting on average all interactions which, in the unit time, can make it switch to whatever else state. One has:
$$ f(t,\,\w)Q^{-}(f)(t,\,\w)=\iint_{I^4}\eta(\w,\,\w^\ast)\cP(\w\to\w'\,\vert\,\w^\ast,\,\rho)f(t,\,\w)f(t,\,\w^\ast)\,d\w'\,d\w^\ast, $$
then property~\eqref{eq:tabgames} gives the above expression for $Q^{-}(f)$.

Putting together the terms introduced so far, and assuming, for the sake of simplicity, that $\eta(\w_\ast,\,\w^\ast)=1$ for all $\w_\ast,\,\w^\ast\in I^2$, we finally obtain the following integro-differential equation for $f$:
\begin{equation}
	\partial_tf=\iint_{I^4}\cP(\w_\ast\to\w\,\vert\,\w^\ast,\,\rho)f(t,\,\w_\ast)f(t,\,\w^\ast)\,d\w_\ast\,d\w^\ast-\rho f.
	\label{eq:kinetic_model}
\end{equation}
Notice that, with constant interaction rates, the loss term is directly proportional to the vehicle density. Moreover, owing to property~\eqref{eq:tabgames}, it results
$$ \int_{I^2}\Bigl(Q^{+}(f,\,f)(t,\,\w)-\rho f(t,\,\w)\Bigr)\,d\w=0,  $$
therefore integrating~\eqref{eq:kinetic_model} with respect to $\w$ gives that such a density is actually constant in time (conservation of mass).

For the study of basic qualitative properties of~\eqref{eq:kinetic_model} the reader may refer to Appendix~\ref{app:basictheo}, where we tackle the well-posedness of the Cauchy problem associated with~\eqref{eq:kinetic_model} in the frame of \emph{measure-valued differential equations} in Wasserstein spaces. Such a theoretical setting, which is more abstract than the one usually considered in the literature for similar equations, see e.g.,~\cite[Appendix A]{bellomo2015NHM} and~\cite{pucci2014DCDSS}, is here motivated by the specialisation of the model that we are going to discuss in the next section.

\subsection{Discrete microscopic states}
\label{sec:discrete}
For practical reasons, it may be convenient to think of the microscopic states $v,\,u$ as \emph{quantised} (i.e., distributed over a set of \emph{discrete}, rather than continuous, values). This is particularly meaningful for the personal risk $u$, which is a non-mechanical quantity naturally meant in levels, but may be reasonable also for the speed $v$, see e.g.,~\cite{delitala2007M3AS,fermo2013SIAP}, considering that the cruise speed of a vehicle tends to be mostly piecewise constant in time, with rapid transitions from one speed level to another.

In the state space $I^2=[0,\,1]\times [0,\,1]\subset\R^2$ we consider therefore a lattice of microscopic states $\{\w_{ij}\}$, with $\w_{ij}=(v_i,\,u_j)$ and, say, $i=1,\,\dots,\,n$, $j=1,\,\dots,\,m$. For instance, if the lattice is uniformly spaced we have
$$ v_i=\frac{i-1}{n-1}, \qquad u_j=\frac{j-1}{m-1}, $$
with in particular $v_1=u_1=0$, $v_n=u_m=1$, and $v_i<v_{i+1}$ for all $i=1,\,\dots,\,n-1$, $u_j<u_{j+1}$ for all $j=1,\,\dots,\,m-1$.

Proceeding at first in a formal fashion, over such a lattice we postulate the following form of the kinetic distribution function:
\begin{equation}
	f(t,\,\w)=\sum_{i=1}^{n}\sum_{j=1}^{m}f_{ij}(t)\delta_{\w_{ij}}(\w),
	\label{eq:discr_f}
\end{equation}
where $\delta_{\w_{ij}}(\w)=\delta_{v_i}(v)\otimes\delta_{u_j}(u)$ is the two-dimensional Dirac delta function, while $f_{ij}(t)$ is the fraction of vehicles which, at time $t$, travel at speed $v_i$ with personal risk $u_j$ (or, depending on the interpretation given to $f$, it is the probability that a representative vehicle of the system possesses the microscopic state $\w_{ij}=(v_i,\,u_j)$ at time $t$). In order to specialise~\eqref{eq:kinetic_model} to the kinetic distribution function~\eqref{eq:discr_f}, we rewrite it in weak form by multiplying by a test function $\phi\in C(I^2)$ and integrating over $I^2$:
\begin{align*}
	\frac{d}{dt} & \int_{I^2}\phi(\w)f(t,\,\w)\,d\w \\
	&= \iint_{I^4}\left(\int_{I^2}\phi(\w)\cP(\w_\ast\to\w\,\vert\,\w^\ast,\,\rho)\,d\w\right)
		f(t,\,\w_\ast)f(t,\,\w^\ast)\,d\w_\ast\,d\w^\ast \\
	&\phantom{=} -\rho\int_{I^2}\phi(\w)f(t,\,\w)\,d\w;
\end{align*}
next we read $f(t,\,\w)\,d\w$ as an integration measure, not necessarily regular with respect to Lebesgue, and from~\eqref{eq:discr_f} we get:
\begin{align*}
	\sum_{i=1}^{n} & \sum_{j=1}^{m}f'_{ij}(t)\phi(\w_{ij}) \\
	&= \sum_{i_\ast,i^\ast=1}^{n}\sum_{j_\ast,j^\ast=1}^{m}
		\left(\int_{I^2}\phi(\w)\cP(\w_{i_\ast j_\ast}\to\w\,\vert\,\w_{i^\ast j^\ast},\,\rho)\,d\w\right)
			f_{i_\ast j_\ast}(t)f_{i^\ast j^\ast}(t) \\
	&\phantom{=} -\sum_{i=1}^{n}\sum_{j=1}^{m}\rho f_{ij}(t)\phi(\w_{ij}).
\end{align*}
In view of the quantisation of the state space, the transition probability distribution must have a structure comparable to~\eqref{eq:discr_f}, i.e., it must be a discrete probability distribution over the post-interaction state $\w$. Hence we postulate:
$$ \cP(\w_{i_\ast j_\ast}\to\w\,\vert\,\w_{i^\ast j^\ast},\,\rho)
	=\sum_{i=1}^{n}\sum_{j=1}^{m}\cP_{i_\ast j_\ast,i^\ast j^\ast}^{ij}(\rho)\delta_{\w_{ij}}(\w), $$
where $\cP_{i_\ast j_\ast,i^\ast j^\ast}^{ij}(\rho)\in [0,\,1]$ is the probability that the vehicle with microscopic state $\w_{i_\ast j_\ast}$ jumps to the microscopic state $\w_{ij}$ because of an interaction with the vehicle with microscopic state $\w_{i^\ast j^\ast}$, given the local traffic congestion $\rho$. Plugging this into the equation above yields
\begin{align*}
	\sum_{i=1}^{n} & \sum_{j=1}^{m}f'_{ij}(t)\phi(\w_{ij}) \\
	&= \sum_{i=1}^{n}\sum_{j=1}^{m}\left(\sum_{i_\ast,i^\ast=1}^{n}
		\sum_{j_\ast,j^\ast=1}^{m}\cP_{i_\ast j_\ast,i^\ast j^\ast}^{ij}(\rho)f_{i_\ast j_\ast}(t)f_{i^\ast j^\ast}(t)
			-\rho f_{ij}(t)\right)\phi(\w_{ij}),
\end{align*}
whence finally, owing to the arbitrariness of $\phi$, we obtain
\begin{equation}
	f'_{ij}=\sum_{i_\ast,i^\ast=1}^{n}\sum_{j_\ast,j^\ast=1}^{m}\cP_{i_\ast j_\ast,i^\ast j^\ast}^{ij}(\rho)f_{i_\ast j_\ast}f_{i^\ast j^\ast}-\rho f_{ij}.
	\label{eq:discr_kinetic_model}
\end{equation}

\begin{remark} \label{rem:prop_discr_model}
This discrete-state kinetic-type equation has been studied in the literature, see e.g.,~\cite{colasuonno2013CM,delitala2007M3AS}. In particular, it has been proved to admit smooth solutions $t\mapsto f_{ij}(t)$, which are unique and nonnegative for prescribed nonnegative initial data $f_{ij}(0)$ and, in addition, preserve the total mass $\sum_{i=1}^{n}\sum_{j=1}^{m}f_{ij}(t)$ in time.
\end{remark}

The arguments above can be made rigorous by appealing to the theory for~\eqref{eq:kinetic_model} developed in Appendix~\ref{app:basictheo}. In particular, we can state the following result:
\begin{theorem} \label{theo:cont-discr_link}
Let the transition probability distribution have the form
$$ \cP(\w_\ast\to\w\,\vert\,\w^\ast,\,\rho)
	=\sum_{i=1}^{n}\sum_{j=1}^{m}\cP^{ij}(\w_\ast,\,\w^\ast,\,\rho)\delta_{\w_{ij}}(\w), $$
where the mapping $(\w_\ast,\,\w^\ast,\,\rho)\mapsto\cP^{ij}(\w_\ast,\,\w^\ast,\,\rho)$ is Lipschitz continuous for all $i,\,j$, i.e., there exists a constant $\Lip(\cP^{ij})>0$ such that
\begin{multline*}
	\abs{\cP^{ij}(\w_{\ast 2},\,\w^\ast_2,\,\rho_2)-\cP^{ij}(\w_{\ast 2},\,\w^\ast_2,\,\rho_2)} \\
		\leq\Lip(\cP^{ij})\Bigl(\abs{\w_{\ast 2}-\w_{\ast 1}}+\abs{\w^\ast_2-\w^\ast_1}+\abs{\rho_2-\rho_1}\Bigr)
\end{multline*}
for all $\w_{\ast 1},\,\w_{\ast 2},\,\w^\ast_1,\,\w^\ast_2\in I^2$, $\rho\in [0,\,1]$.

Let moreover
$$ f_0(\w)=\sum_{i=1}^{n}\sum_{j=1}^{m}f^0_{ij}\delta_{\w_{ij}}(\w) $$
be a prescribed kinetic distribution function at time $t=0$ over the lattice of microscopic states $\{\w_{ij}\}\subset I^2$, such that
$$ f^0_{ij}\geq 0\ \forall\,i,\,j, \qquad \sum_{i=1}^{n}\sum_{j=1}^{m}f^0_{ij}=\rho\in [0,\,1]. $$

Set $\cP_{i_\ast j_\ast,i^\ast j^\ast}^{ij}(\rho):=\cP^{ij}(\w_{i_\ast j_\ast},\,\w_{i^\ast j^\ast},\,\rho)$. Then the corresponding unique solution to~\eqref{eq:kinetic_model} is~\eqref{eq:discr_f} with coefficients $f_{ij}(t)$ given by~\eqref{eq:discr_kinetic_model} along with the initial conditions $f_{ij}(0)=f^0_{ij}$. In addition, it depends continuously on the initial datum as stated by Theorem~\ref{theo:cont_dep}.
\end{theorem}
\begin{proof}
The given transition probability distribution satisfies Assumption~\ref{ass:Lip_P}, in fact
\begin{align*}
	& \wass{\cP(\w_{\ast 1}\to\cdot\,\vert\,\w^\ast_1,\,\rho_1)}{\cP(\w_{\ast 2}\to\cdot\,\vert\,\w^\ast_2,\,\rho_2)} \\
	&\qquad =\sup_{\varphi\in C_{b,1}(I^2)\cap\Lip_1(I^2)}\sum_{i=1}^{n}\sum_{j=1}^{m}\varphi(\w_{ij})
		\left(\cP^{ij}(\w_{\ast 2},\,\w^\ast_2,\,\rho_2)-\cP^{ij}(\w_{\ast 1},\,\w^\ast_1,\,\rho_1)\right) \\
	&\qquad \leq\left(\sum_{i=1}^{n}\sum_{j=1}^{m}\Lip(\cP^{ij})\right)
		\Bigl(\abs{\w_{\ast 2}-\w_{\ast 1}}+\abs{\w^\ast_2-\w^\ast_1}+\abs{\rho_2-\rho_1}\Bigr).
\end{align*}
Furthermore, $f_0\in\cM{\rho}$. Then, owing to Theorem~\ref{theo:exist_uniqueness}, we can assert that the Cauchy problem associated with~\eqref{eq:kinetic_model} admits a unique mild solution\footnote{For the definition of mild solution to~\eqref{eq:kinetic_model} we refer the reader to~\eqref{eq:mild} in Appendix~\ref{app:basictheo}.}.

The calculations preceding this theorem show that if the $f_{ij}(t)$'s satisfy~\eqref{eq:discr_kinetic_model} then~\eqref{eq:discr_f} is indeed such a solution, considering also that it is nonnegative (cf. Remark~\ref{rem:prop_discr_model}) and matches the initial condition $f_0$.
\end{proof}

Theorem~\ref{theo:cont-discr_link} requires the mapping $(\w_\ast,\,\w^\ast,\,\rho)\mapsto\cP^{ij}(\w_\ast,\,\w^\ast,\,\rho)$ to be Lipschitz continuous in $I^2\times I^2\times [0,\,1]$ but the solution~\eqref{eq:discr_f} depends ultimately only on the values $\cP^{ij}(\w_{i_\ast j_\ast},\,\w_{i^\ast j^\ast},\,\rho)$, cf.~\eqref{eq:discr_kinetic_model}. Therefore, when constructing specific models, we can confine ourselves to specifying the values $\cP_{i_\ast j_\ast,i^\ast j^\ast}^{ij}(\rho)$, taking for granted that they can be variously extended to points $(\w_\ast,\,\w^\ast)\ne (\w_{i_\ast j_\ast},\,\w_{i^\ast j^\ast})$ in a Lipschitz continuous way.

\section{Modelling microscopic interactions}
\label{sec:interactions}
From now on, we will systematically refer to the discrete-state setting ruled by~\eqref{eq:discr_kinetic_model}. In order to describe the interactions among the vehicles, it is necessary to model the transition probabilities $\cP_{i_\ast j_\ast,i^\ast j^\ast}^{ij}(\rho)$ associated with the jump processes over the lattice of discrete microscopic states.

As a first step, we propose the following factorisation:
\begin{align*}
	\cP_{i_\ast j_\ast,i^\ast j^\ast}^{ij}(\rho) &=
		\Prob(u_{j_\ast}\to u_j\,\vert\,v_{i_\ast},\,v_{i^\ast},\,\rho)\cdot\Prob(v_{i_\ast}\to v_i\,\vert\,v_{i^\ast},\,\rho) \\
	& =: (\cP')_{i_\ast j_\ast,i^\ast}^{j}(\rho)\cdot(\cP'')_{i_\ast i^\ast}^{i}(\rho),
\end{align*}
which implies that changes in the personal risk (first term at the right-hand side) depend on the current speeds of the interacting pairs while the driving style (i.e., the way in which the speed changes, second term at the right-hand side) is not directly influenced by the current personal risk. In a sense, we are interpreting the change of personal risk as a function of the driving conditions, however linked to the \emph{subjectivity} of the drivers and hence described in probability. By subjectivity we mean the fact that different drivers may not respond in the same way to the same conditions. More advanced models may account for a joint influence of speed and risk levels on binary interactions, but for the purposes of the present paper the approximation above appears to be satisfactory.

As a second step, we detail the transition probabilities $(\cP')_{i_\ast j_\ast,i^\ast}^{j}(\rho)$, $(\cP'')_{i_\ast i^\ast}^{i}(\rho)$ just introduced. It is worth stressing that they will be mainly inspired by a prototypical analysis of the driving styles. In particular, they will be parameterised by the vehicle density $\rho\in [0,\,1]$ so as to feed back the global traffic conditions to the local interaction rules. Other external \emph{objective} factors which may affect the flow of vehicles and the personal risk, such as e.g., weather or road conditions (number of lanes, number of directions of travel, type of wearing course), will be summarised by a parameter $\alpha\in [0,\,1]$, whose low, resp. high, values stand for poor, resp. good, conditions.

\begin{remark}
The interpretation of $\alpha$ is conceptually analogous to that of $u$ discussed in Remark~\ref{rem:u}: its numerical values do not refer to actual physical (measured) ranges but serve to convey, in mathematical terms, the influence of external conditions on binary interactions.
\end{remark}

\subsection{Risk transitions}
\label{sec:risk_trans}
In modelling the risk transition probability
$$ (\cP')_{i_\ast j_\ast,i^\ast}^{j}(\rho)=\Prob(u_{j_\ast}\to u_j\,\vert\,v_{i_\ast},v_{i^\ast},\,\rho) $$
we consider two cases, depending on whether the vehicle with state $(v_{i_\ast},\,u_{j_\ast})$ interacts with a faster or a slower leading vehicle with speed $v_{i^\ast}$.
\begin{itemize}
\item If $v_{i_\ast}\leq v_{i^\ast}$ we set
$$ (\cP')_{i_\ast j_\ast,i^\ast}^{j}(\rho)=\alpha\rho\delta_{j,\max\{1,\,j_\ast-1\}}+(1-\alpha\rho)\delta_{j,j_\ast}, $$
the symbol $\delta$ denoting here the Kronecher's delta. In practice, we assume that the interaction with a faster leading vehicle can reduce the personal risk with probability $\alpha\rho$, which raises in high traffic congestion and good environmental conditions. The rationale is that the headway from a faster leading vehicle increases, which reduces the risk of collision especially when vehicles are packed (high $\rho$) or when speeds are presumably high (good environmental conditions, i.e., high $\alpha$). Alternatively, after the interaction the personal risk remains the same with the complementary probability.
\item If $v_{i_\ast}>v_{i^\ast}$ we set
$$ (\cP')_{i_\ast j_\ast,i^\ast}^{j}(\rho)=\delta_{j,\min\{j_\ast+1,\,m\}}, $$
i.e., we assume that the interaction with a slower leading vehicle can only increase the personal risk because the headway is reduced or overtaking is induced (see below).
\end{itemize}

\subsection{Speed transitions}
\label{sec:speed_trans}
In modelling the speed transition probability
$$ (\cP'')_{i_\ast i^\ast}^{i}(\rho)=\Prob(v_{i_\ast}\to v_i\,\vert\,v_{i^\ast},\,\rho) $$
we refer to~\cite{puppo2016CMS}, where the following three cases are considered:
\begin{itemize}
\item If $v_{i_\ast}<v_{i^\ast}$ then
$$ (\cP'')_{i_\ast i^\ast}^{i}(\rho)=\alpha(1-\rho)\delta_{i,i_\ast+1}+(1-\alpha(1-\rho))\delta_{i,i_\ast}, $$
i.e., the vehicle with speed $v_{i_\ast}$ emulates the leading one with speed $v_{i^\ast}$ by accelerating to the next speed with probability $\alpha(1-\rho)$. This probability increases if environmental conditions are good and traffic is not too much congested. Otherwise, the speed remains unchanged with complementary probability.
\item If $v_{i_\ast}>v_{i^\ast}$ then
$$ (\cP'')_{i_\ast i^\ast}^{i}(\rho)=\alpha(1-\rho)\delta_{i,i_\ast}+(1-\alpha(1-\rho))\delta_{i,i^\ast}, $$
i.e., the vehicle with speed $v_{i_\ast}$ maintains its speed with probability $\alpha(1-\rho)$. The rationale is that if environmental conditions are good enough or traffic is sufficiently uncongested then it can overtake the slower leading vehicle with speed $v_{i^\ast}$. Otherwise, it is forced to slow down to the speed $v_{i^\ast}$ and to queue up, which happens with the complementary probability.
\item If $v_{i_\ast}=v_{i^\ast}$ then
\begin{align*}
	(\cP'')_{i_\ast i^\ast}^{i}(\rho) &= (1-\alpha)\rho\delta_{i,\max\{1,\,i_\ast-1\}}+\alpha(1-\rho)\delta_{i,\min\{i_\ast+1,\,n\}} \\
	& \phantom{=} +(1-\alpha-(1-2\alpha)\rho)\delta_{i,i_\ast}.
\end{align*}
In this case there are three possible outcomes of the interaction: if environmental conditions are poor and traffic is congested the vehicle with speed $v_{i_\ast}$ slows down with probability $(1-\alpha)\rho$; if, instead, environmental conditions are good and traffic is light then it accelerates to the next speed (because e.g., it overtakes the leading vehicle) with probability $\alpha(1-\rho)$; finally, it can also remain with its current speed with a probability which complements the sum of the previous two.
\end{itemize}

\section{Case studies}
\label{sec:simulations}
\subsection{Fundamental diagrams of traffic}
\label{sec:funddiag}
Model~\eqref{eq:discr_kinetic_model} can be used to investigate the long-term macroscopic dynamics resulting from the small-scale interactions among vehicles discussed in the previous section. Such dynamics are summarised by the well-known \emph{fundamental} and \emph{speed diagrams} of traffic, see e.g.,~\cite{li2011TRR}, which express the average flux and mean speed of the vehicles at equilibrium, respectively, as functions of the vehicle density along the road. This information, typically obtained from experimental measurements~\cite{bonzani2003MCM,kerner2004BOOK}, is here studied at a theoretical level in order to discuss qualitatively the impact of the driving style on the macroscopically observable traffic trends.

In Appendix~\ref{app:basictheo} we give sufficient conditions for the existence, uniqueness, and global attractiveness of equilibria $f_\infty\in\cM{\rho}$ of~\eqref{eq:kinetic_model}, cf. Theorems~\ref{theo:equilibria},~\ref{theo:attractiveness}. Here we claim, in particular, that if the transition probability distribution $\cP$ has the special form discussed in Theorem~\ref{theo:cont-discr_link} then $f_\infty$ is actually a discrete-state distribution function.

\begin{theorem} \label{theo:discr_equilibria}
Fix $\rho\in [0,\,1]$ and let $\cP(\w_\ast\to\w\,\vert\,\w^\ast,\,\rho)$ be like in Theorem~\ref{theo:cont-discr_link}. Assume moreover that $\Lip(\cP)<\frac{1}{2}$ (cf. Assumption~\ref{ass:Lip_P}). Then the unique equilibrium distribution function $f_\infty\in\cM{\rho}$, which is also globally attractive, has the form $f_\infty(\w)=\sum_{i=1}^{n}\sum_{j=1}^{m}f^\infty_{ij}\delta_{\w_{ij}}(\w)$ with the coefficients $f^\infty_{ij}$ satisfying
\begin{equation}
	f^\infty_{ij}=\frac{1}{\rho}\sum_{i_\ast,i^\ast=1}^{n}\sum_{j_\ast,j^\ast=1}^{m}\cP_{i_\ast j_\ast,i^\ast j^\ast}^{ij}(\rho)
		f^\infty_{i_\ast j_\ast}f^\infty_{i^\ast j^\ast},
	\quad
	\begin{array}{l}
		i=1,\,\dots,\,n \\
		j=1,\,\dots,\,m.
	\end{array}
	\label{eq:discr_equilibria}
\end{equation}
\end{theorem}
\begin{proof}
We consider directly the case $\rho>0$, for $\rho=0$ implies uniquely $f_\infty\equiv 0$.

Since $\Lip(\cP)<\frac{1}{2}$, we know from Theorems~\ref{theo:equilibria},~\ref{theo:attractiveness} that~\eqref{eq:kinetic_model} admits a unique and globally attractive equilibrium distribution $f_\infty\in\cM{\rho}$, which is found as the fixed point of the mapping $f\mapsto\frac{1}{\rho}Q^{+}(f,\,f)$. In particular, defining the subset
$$ D:=\left\{f\in\cM{\rho}\,:\,f(\w)=\sum_{i=1}^{n}\sum_{j=1}^{m}f_{ij}\delta_{\w_{ij}}(\w),\
	f_{ij}\geq 0,\ \sum_{i=1}^{n}\sum_{j=1}^{m}f_{ij}=\rho\right\}, $$
it is easy to see that if $\cP$ has the form indicated in Theorem~\ref{theo:cont-discr_link} then the operator $\frac{1}{\rho}Q^{+}$ maps $D$ into itself. In fact, for $f\in D$ we get
\begin{equation*}
	\frac{1}{\rho}Q^{+}(f,\,f)=\sum_{i=1}^{n}\sum_{j=1}^{m}\left(\frac{1}{\rho}\sum_{i_\ast,i^\ast=1}^{n}\sum_{j_\ast,j^\ast=1}^{m}
		\cP_{i_\ast j_\ast,i^\ast j^\ast}^{ij}(\rho)f_{i_\ast j_\ast}f_{i^\ast j^\ast}\right)\delta_{\w_{ij}}(\w).
\end{equation*}
From here we also deduce formally~\eqref{eq:discr_equilibria}. Therefore, in order to get the thesis, it is sufficient to prove that $D$ is closed in $\cM{\rho}$. In fact this will imply that $(D,\,W_1)$ is a complete metric space, and Banach contraction principle will then locate the fixed point of $\frac{1}{\rho}Q^{+}$ in $D$.

Let $(f^k)\subseteq D$ be a convergent sequence in $\cM{\rho}$ with respect to the $W_1$ metric. It is then Cauchy, hence given $\epsilon>0$ we find $N_\epsilon\in\mathbb{N}$ such that if $h,\,k>N_\epsilon$ then $\wass{f^h}{f^k}<\epsilon$. This condition means
$$ \abs{\int_{I^2}\varphi(\w)\left(f^k(\w)-f^h(\w)\right)\,d\w}=
	\abs{\sum_{i=1}^{n}\sum_{j=1}^{m}\varphi(\w_{ij})\left(f^k_{ij}-f^h_{ij}\right)}<\epsilon $$
for every $\varphi\in C_{b,1}(I^2)\cap\Lip_1(I^2)$. In particular, taking a function $\varphi$ which vanishes at every $\w_{ij}$ but one, say $\w_{\bar{\imath}\bar{\jmath}}$, we discover $\abs{\varphi(\w_{\bar{\imath}\bar{\jmath}})}\abs{f^h_{\bar{\imath}\bar{\jmath}}-f^k_{\bar{\imath}\bar{\jmath}}}<\epsilon$ for all $h,\,k>N_\epsilon$. Thus we deduce that $(f^k_{ij})_k$ is a Cauchy sequence in $\R$, hence for all $i,\,j$ there exists $f_{ij}\in\R$ such that $f^k_{ij}\to f_{ij}$ ($k\to\infty$). Clearly $f_{ij}\geq 0$ because the $f^k_{ij}$'s are all non-negative by assumption; moreover, $\sum_{i=1}^{n}\sum_{j=1}^{m}f_{ij}=\lim_{k\to\infty}\sum_{i=1}^{n}\sum_{j=1}^{m}f^k_{ij}=\rho$. Therefore $f:=\sum_{i=1}^{n}\sum_{j=1}^{m}f_{ij}\delta_{\w_{ij}}\in D$.

We now claim that $f^k\to f$ in the $W_1$ metric:
\begin{align*}
	\wass{f^k}{f} &= \sup_{\varphi\in C_{b,1}(I^2)\cap\Lip_1(I^2)}\int_{I^2}\varphi(\w)\left(f(\w)-f^k(\w)\right)\,d\w \\
	&= \sup_{\varphi\in C_{b,1}(I^2)\cap\Lip_1(I^2)}\sum_{i=1}^{n}\sum_{j=1}^{m}\varphi(\w_{ij})\left(f_{ij}-f^k_{ij}\right) \\
	&\leq \sum_{i=1}^{n}\sum_{j=1}^{m}\abs{f_{ij}-f^k_{ij}}\xrightarrow{k\to\infty}0.
\end{align*}
This implies that $D$ is closed and the proof is completed.
\end{proof}

Under the assumptions of Theorem~\ref{theo:discr_equilibria},~\eqref{eq:discr_equilibria} defines a mapping $[0,\,1]\ni\rho\mapsto\{f^\infty_{ij}(\rho)\}$, i.e., for every $\rho\in [0,\,1]$ there exist unique coefficients $f^\infty_{ij}$ solving~\eqref{eq:discr_equilibria} such that $\{f^\infty_{ij}\}$ is the equilibrium of system~\eqref{eq:discr_kinetic_model} with moreover $\sum_{i=1}^{n}\sum_{j=1}^{m}f^\infty_{ij}=\rho$.

Owing to the argument above, for each $\rho$ it is possible to compute the corresponding average flux and mean speed at equilibrium by means of formulas~\eqref{eq:ave_flux},~\eqref{eq:mean_speed} with the kinetic distribution $f_\infty$. This generates the two mappings
\begin{equation}
	\rho\mapsto q(\rho):=\sum_{i=1}^{n}v_i\sum_{j=1}^{m}f^\infty_{ij}(\rho),
		\qquad \rho\mapsto V(\rho):=\frac{q(\rho)}{\rho},
	\label{eq:funddiag}
\end{equation}
which are the theoretical definitions of the fundamental and speed diagrams, respectively, of traffic. Furthermore, it is possible to estimate the dispersion of the microscopic speeds at equilibrium by computing the standard deviation of the speed:
$$ \sigma_V(\rho):=\sqrt{\frac{1}{\rho}\sum_{i=1}^{n}{\left(v_i-V(\rho)\right)}^2\sum_{j=1}^{m}f^\infty_{ij}(\rho)}, $$
that of the flux being $\rho\sigma_V(\rho)$, which gives a measure of the homogeneity of the driving styles of the drivers.

Figure~\ref{fig:funddiag} shows the diagrams~\eqref{eq:funddiag}, with the corresponding standard deviations, for different values of the constant $\alpha$ parameterising the transition probabilities, cf. Sections~\ref{sec:risk_trans},~\ref{sec:speed_trans}. Each pair $(\rho,\,q(\rho))$, $(\rho,\,V(\rho))$ has been computed by integrating numerically~\eqref{eq:discr_kinetic_model} up to a sufficiently large final time, such that the equilibrium was reached.

\begin{figure}[!t]
\includegraphics[width=0.9\textwidth]{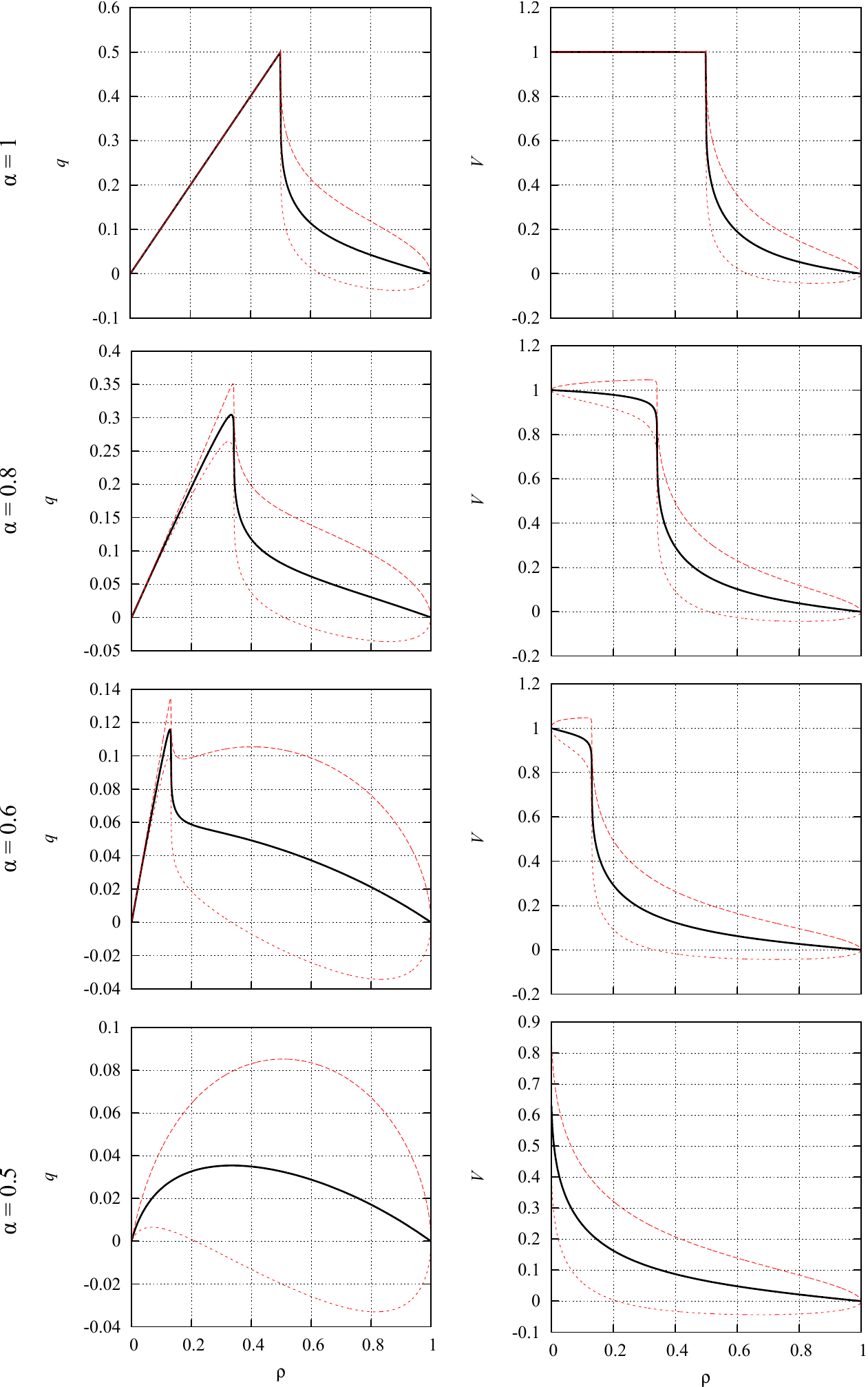}
\caption{Average flux (left column) and mean speed (right column) vs. traffic density for different levels of the quality of the environment. Red dashed lines are the respective standard deviations. $n=6$ uniformly spaced speed classes have been used.}
\label{fig:funddiag}
\end{figure}

For $\alpha=1$ (best environmental conditions) the diagrams are the same as those studied analytically in~\cite{fermo2014DCDSS}. In particular, they show a clear separation between the so-called \emph{free} and \emph{congested phases} of traffic: the former, taking place at low density ($\rho<0.5$), is characterised by the fact that vehicles travel at the maximum speed with zero standard deviation; the latter, taking place instead at high density ($\rho>0.5$), is characterised by a certain dispersion of the microscopic speeds with respect to their mean value. In~\cite{fermo2014DCDSS} the critical value $\rho=0.5$ has been associated with a supercritical bifurcation of the equilibria, thereby providing a precise mathematical characterisation of the phase transition.

For $\alpha<1$, analogously detailed analytical results are not yet available and, to our knowledge, Theorems~\ref{theo:equilibria},~\ref{theo:attractiveness} in Appendix~\ref{app:basictheo} are the first results giving at least sufficient conditions for the qualitative characterisation of equilibrium solutions to~\eqref{eq:kinetic_model} in the general case. According to the graphs in Figure~\ref{fig:funddiag}, the model predicts lower and lower critical values for the density threshold triggering the phase transition. In addition to that, coherently with the experimental observations, cf. e.g.,~\cite{kerner2004BOOK}, some scattering of the diagrams appears also for low density, along with a \emph{capacity drop} visible in the average flux (i.e., the fact that the maximum flux in the congested phase is lower than the maximum one in the free phase), which separates the free and congested phases as described in~\cite{zhang2005TRB}.

Compared to typical experimental data, the most realistic diagrams seem to be those obtained for $\alpha=0.8$, which denotes suboptimal though not excessively poor environmental conditions. It is worth stressing that such a realism of the theoretical diagrams is not only relevant for supporting the derivation of macroscopic kinematic features of the flow of vehicles at equilibrium out of microscopic interaction rules far from equilibrium. It constitutes also a reliable basis for interpreting, in a similar multiscale perspective, the link with risk and safety issues, for which synthetic and informative empirical data similar to the fundamental and speed diagrams are, to the authors knowledge, not currently available for direct comparison.

\subsection{Risk diagrams of traffic}
\label{sec:riskdiag}
Starting from the statistical distribution of the risk given in~\eqref{eq:risk_distr}, we propose the following definition for the probability of accident along the road:
\begin{definition}[Probability of accident] \label{def:p_acc}
Let $\bar{u}\in (0,\,1)$ be a risk threshold above which the personal risk for a representative vehicle is considered too high. We define the instantaneous \emph{probability of accident} $P=P(t)$ (associated with $\bar{u}$) as the normalised number of vehicles whose personal risk is, at time $t$, greater than or equal to $\bar{u}$:
$$ P(t):=\frac{1}{\rho}\int_{\bar{u}}^1\varphi(t,\,u)\,du=\frac{1}{\rho}\int_{\bar{u}}^1\int_0^1 f(t,\,v,\,u)\,dv\,du. $$
\end{definition}

Asymptotically, using the discrete-state equilibrium distribution function $f_\infty=f_\infty(\rho)$ found in Theorem~\ref{theo:discr_equilibria}, we obtain the mapping
\begin{equation}
	\rho\mapsto P(\rho):=\frac{1}{\rho}\sum_{j\,:\,u_j\geq\bar{u}}\sum_{i=1}^{n}f^\infty_{ij}(\rho),
	\label{eq:p_acc}
\end{equation}
which shows that the probability of accident is, in the long run, a function of the traffic density. We call~\eqref{eq:p_acc} the \emph{accident probability diagram}.

Definition~\ref{def:p_acc} and, in particular,~\eqref{eq:p_acc} depend on the threshold $\bar{u}$, which needs to be estimated in order for the model to serve quantitative purposes. If, for a given road, the empirical probability of accident is known (for instance, from time series on the frequency of accidents, see e.g.,~\cite{abdel-aty2000AAP,miaou1993AAP,oppe1989AAP}) then it is possible to find $\bar{u}$  by solving an inverse problem which leads the theoretical probability~\eqref{eq:p_acc} to match the experimental one. This way, the road under consideration can be assigned the \emph{risk threshold} $\bar{u}$.

The question then arises how to use the information provided by the risk threshold $\bar{u}$ for the assessment of safety standards. In fact, the personal risk $u$, albeit a primitive variable of the model, is not a quantity which can be really measured for each vehicle: a macroscopic synthesis is necessary. Quoting from~\cite{KiwiRAP2012}:
\begin{quote}
Personal risk is most of interest to the public, as it shows the risk to road users, as individuals.
\end{quote}
To this purpose, we need to further post-process the statistical information brought by the kinetic model. Taking inspiration from the fundamental and speed diagrams of traffic discussed in Section~\ref{sec:funddiag}, a conceivable approach is to link the personal risk, conveniently understood in an average sense, to the macroscopic traffic density along the road. For this we define:
\begin{definition}[Risk diagram] \label{def:riskdiag}
The \emph{risk diagram} of traffic is the mapping (cf.~\eqref{eq:ave_risk})
\begin{equation}
	\rho\mapsto U(\rho):=\frac{1}{\rho}\sum_{j=1}^{m}u_j\sum_{i=1}^{n}f^\infty_{ij}(\rho),
	\label{eq:riskdiag}
\end{equation}
$f_\infty=f_\infty(\rho)$ being the equilibrium kinetic distribution function of Theorem~\ref{theo:discr_equilibria}. The related standard deviation is
$$ \sigma_U(\rho):=\sqrt{\frac{1}{\rho}\sum_{j=1}^{m}{\left(u_j-U(\rho)\right)}^2\sum_{i=1}^{n}f^\infty_{ij}(\rho)}. $$
\end{definition}

Using the tools provided by Definition~\ref{def:riskdiag}, we can finally fix a \emph{safety criterion} which discriminates between safety and risk regimes of traffic depending on the traffic loads:
\begin{definition}[Safety criterion] \label{def:safety_crit}
Let $\bar{u}\in (0,\,1)$ be the risk threshold fixed by Definition~\ref{def:p_acc}. The \emph{safety regime} of traffic along a given road corresponds to the traffic loads $\rho\in [0,\,1]$ such that
$$ U(\rho)+\sigma_U(\rho)<\bar{u}. $$
The complementary regime, i.e., the one for which $U(\rho)+\sigma_U(\rho)\geq\bar{u}$, is the \emph{risk regime}.
\end{definition}

\begin{remark}
An alternative, less precautionary, criterion might identify the safety regime with the traffic loads such that $U(\rho)<\bar{u}$ and the risk regime with those such that $U(\rho)\geq\bar{u}$.
\end{remark}

\begin{figure}[!t]
\includegraphics[width=0.9\textwidth]{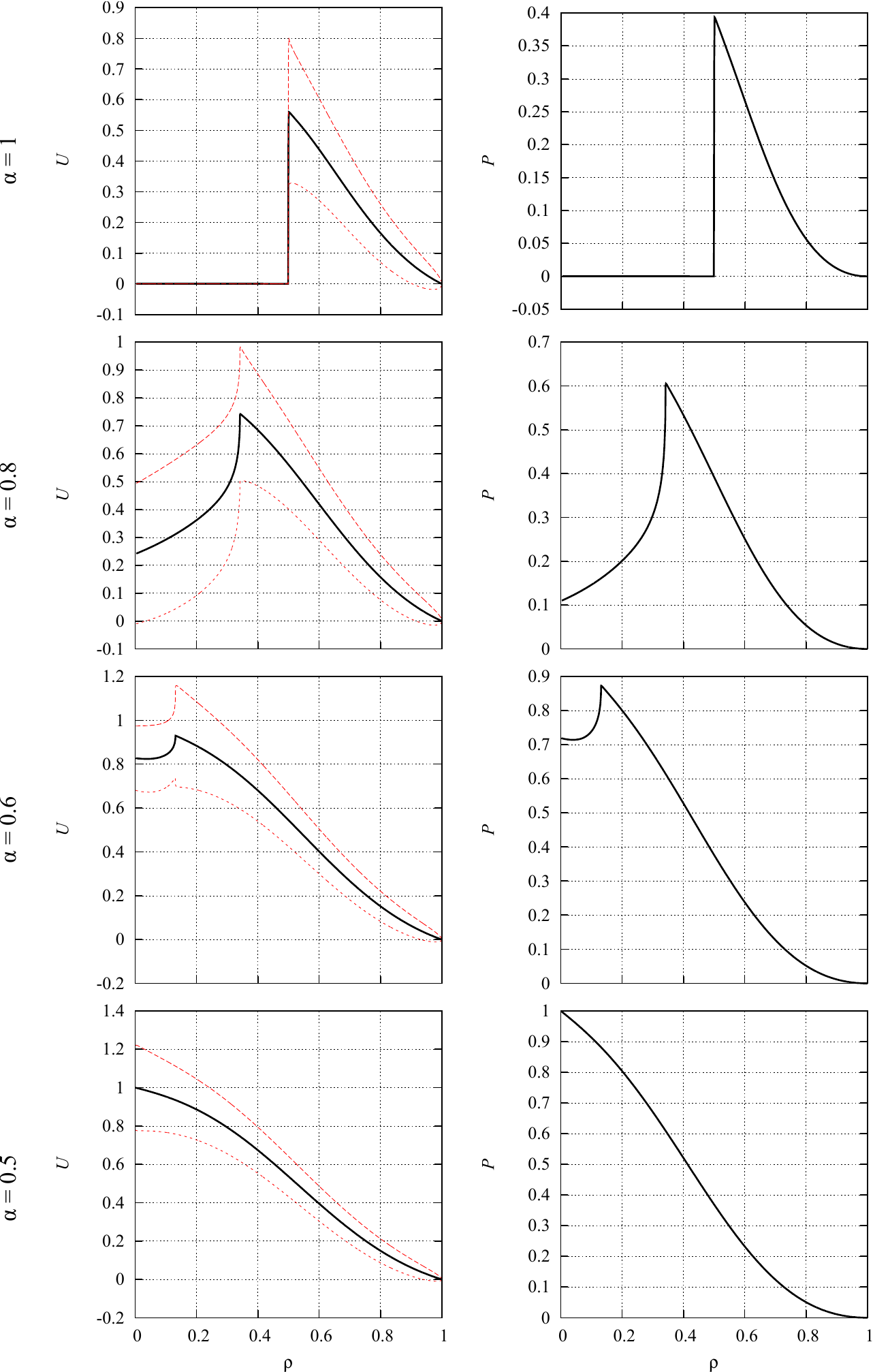}
\caption{Average risk (left column), with related standard deviation (red dashed lines), and probability of accident (right column) vs. traffic density for different levels of the quality of the environment. $n=6$ and $m=3$ uniformly spaced speed and risk classes, respectively, have been used}
\label{fig:riskdiag}
\end{figure}

Figure~\ref{fig:riskdiag} shows the risk diagram~\eqref{eq:riskdiag} and the probability of accident~\eqref{eq:p_acc} at equilibrium obtained numerically for various environmental conditions as described in Section~\ref{sec:funddiag}, using the heuristic risk threshold $\bar{u}=0.7$.

For $\alpha=1$ the average risk and the probability of accident are deterministically zero for all values of the traffic density in the free phase ($\rho<0.5$). This is a consequence of the fact that, as shown by the corresponding diagrams in Figure~\ref{fig:funddiag}, at low density in optimal environmental conditions vehicles virtually do not interact, all of them travelling undisturbed at the maximum speed. In practice, vehicles behave as if the road were empty and consequently the model predicts maximal safety with no possibility of collisions. In the congested phase ($\rho>0.5$), instead, the average risk and the probability of accident rise suddenly to a positive value, following the emergence of the scattering of the microscopic speeds, see again the corresponding panels in Figure~\ref{fig:funddiag}. Then they decrease monotonically to zero when $\rho$ approaches $1$, for in a full traffic jam vehicles do not move. Notice that the maximum of the average risk and of the probability of accident is in correspondence of the critical density value $\rho=0.5$.

For $\alpha<1$, the main features of the diagrams $\rho\mapsto U(\rho)$ and $\rho\mapsto P(\rho)$ described in the ideal prototypical case above remain unchanged. In particular, a comparison with Figure~\ref{fig:funddiag} shows that the maximum of both diagrams is still reached in correspondence of the density value triggering the transition from free to congested traffic, see also Figure~\ref{fig:postproc}a. This is clearly in good agreement with the intuition, indeed it identifies the phase transition as the most risky situation for drivers. It is worth remarking that such a macroscopically observable fact has not been postulated in the construction of the model but has emerged as a result of more elementary microscopic interaction rules. For $\alpha=0.6,\,0.8$, the average risk and the probability of accident take, in the free phase of traffic, a realistic nonzero value, which first increases before the phase transition and then decreases to zero in the congested phase. Notice that the standard deviation of the risk is higher in the free than in the congested phase, which again meets the intuition considering that at low density the movement of single vehicles is less constrained by the global flow. For $\alpha=0.5$, environmental conditions are so poor that $U(\rho),\,P(\rho)\to 1$ when $\rho\to 0^+$, namely in correspondence of the maximum mean speed (cf. the corresponding panels of Figure~\ref{fig:funddiag}).

\begin{figure}[!t]
\includegraphics[width=\textwidth]{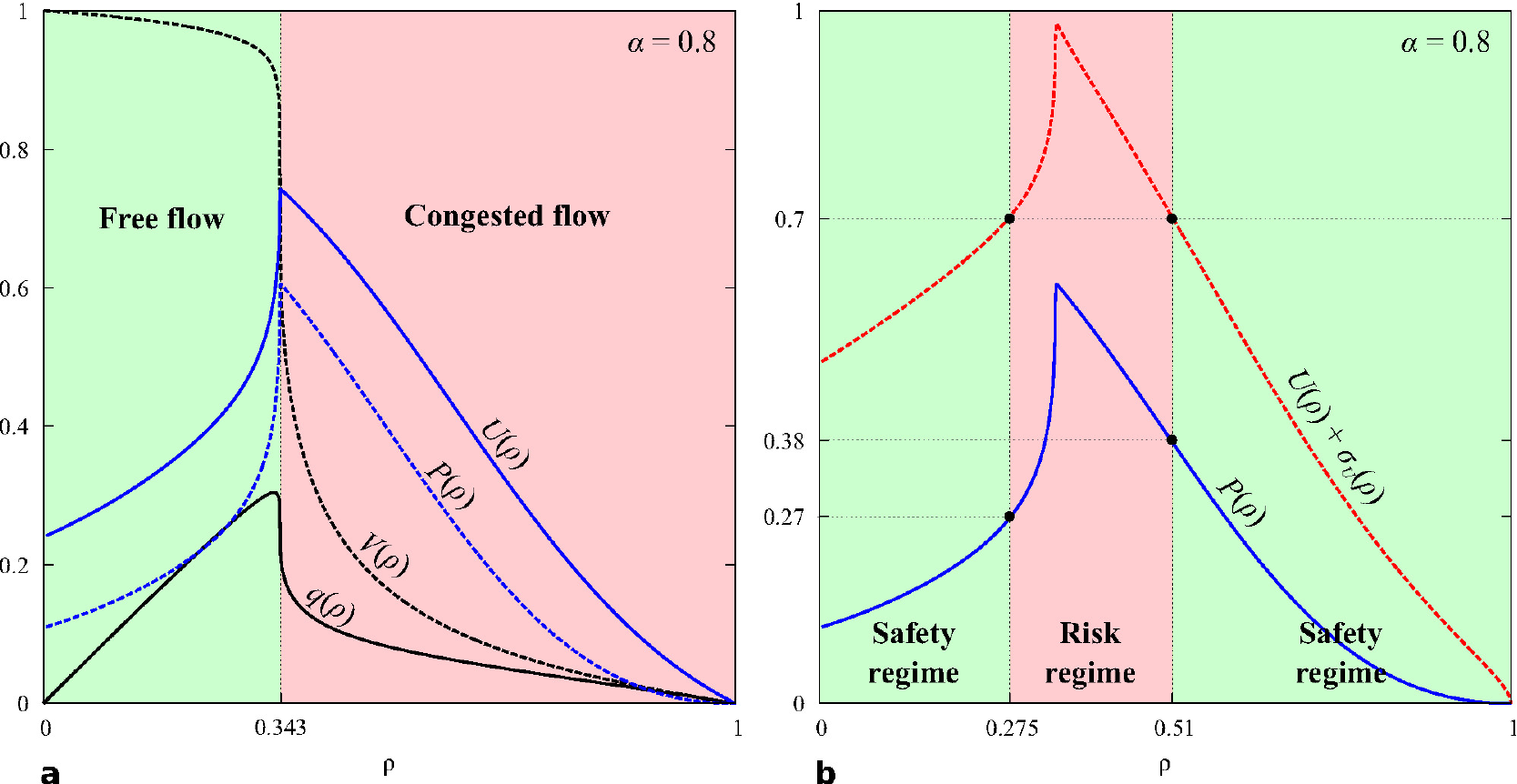}
\caption{\textbf{a.}~Comparison of the curves~\eqref{eq:funddiag},~\eqref{eq:p_acc},~\eqref{eq:riskdiag} in the case $\alpha=0.8$. \textbf{b.}~Determination of safety and risk regimes of traffic for $\alpha=0.8$, with corresponding probabilities of accident.}
\label{fig:postproc}
\end{figure}

According to~\cite{KiwiRAP2012}:
\begin{quote}
Personal risk is typically higher in more difficult terrains, where traffic volumes and road standards are often lower.
\end{quote}
By looking at the graphs in the first column of Figure~\ref{fig:riskdiag}, we see that the results of the model match qualitatively well this experimental observation: for low $\rho$ and decreasing $\alpha$ the average personal risk tends indeed to increase.

Again, the most realistic (namely suboptimal, though not excessively poor) scenario appears to be the one described by $\alpha=0.8$. In this case, cf. Figure~\ref{fig:postproc}b, the safety criterion of Definition~\ref{def:safety_crit}, i.e.,
$$ U(\rho)+\sigma_U(\rho)<0.7, $$
individuates two safety regimes for traffic loads $\rho\in [0,\,\rho_1)$ and $\rho\in (\rho_2,\,1]$, with $\rho_1\approx 0.275$ and $\rho_2\approx 0.51$, respectively. The first one corresponds to a probability of accident $P(\rho)\lesssim 27\%$, the second one to $P(\rho)\lesssim 38\%$. Not surprisingly, the maximum admissible probability of accident in free flow ($\rho<\rho_1$) is lower than the one in congested flow ($\rho>\rho_2$), meaning that the safety criterion of Definition~\ref{def:safety_crit} turns out to be more restrictive in the first than in the second case. This can be understood thinking of the fact that in free flow speeds are higher and the movement of vehicles is less constrained by the global flow, which imposes tighter safety standards.

\section{Conclusions and perspectives}
\label{sec:conclusions}
In this paper we have proposed a Boltzmann-type kinetic model which describes the influence of the driving style on the personal driving risk in terms of microscopic binary interactions among the vehicles. In particular, speed transitions due to encounters with other vehicles, and the related changes of personal risk, are described in probability, thereby accounting for the interpersonal variability of the human behaviour, hence ultimately for the \emph{subjective} component of the risk. Moreover, they are parameterised by the local density of vehicles along the road and by the environmental conditions (for instance, type of road, weather conditions), so as to include in the mathematical description also the \emph{objective} component of the risk.

By studying the equilibrium solutions of the model, we have defined two macroscopic quantities of interest for the global assessment of the risk conditions, namely the \emph{risk diagram} and the \emph{accident probability diagram}. The former gives the average risk along the road and the latter the probability of accident both as functions of the density of vehicles, namely of the level of traffic congestion. These diagrams compare well with the celebrated fundamental and speed diagrams of traffic, also obtainable from the equilibrium solutions of our kinetic model, in that they predict the maximum risk across the phase transition from free to congested flow, when several perturbative phenomena are known to occur in the macroscopic hydrodynamic behaviour of traffic (such as e.g., capacity drop~\cite{zhang2005TRB}, scattering of speed and flux and appearance of a third phase of ``synchronised flow''~\cite{kerner2004BOOK}). Moreover, within the free and congested regimes they are in good agreement with the experimental findings of accident data collection campaigns: for instance, they predict that the personal risk rises in light traffic and poor environmental conditions, coherently with what is stated e.g., in~\cite{KiwiRAP2012}.

By using the aforesaid diagrams we have proposed the definition of a \emph{safety criterion}, which, upon assigning to a given road a risk threshold based on the knowledge of real data on accidents, individuates \emph{safety} and \emph{risk regimes} depending on the volume of traffic. Once again, it turns out that the risk regime consists of a range of vehicle densities encompassing the critical one at which phase transition occurs. This type of information is perhaps more directly useful to the public than to traffic controlling authorities, because it shows the average risk that a representative road user is subject to. Nevertheless, by identifying traffic loads which may pose safety threats, it also indicates which densities should be preferably avoided along the road and when risk reducing measures should be activated.

This work should be considered as a very first attempt to formalise, by a mathematical model, the risk dynamics in vehicular traffic from the point of view of \emph{simulation} and \emph{prediction} rather than simply of statistical description. Several improvements and developments are of course possible, which can take advantage of some existing literature about kinetic models of vehicular traffic: for instance, one may address the spatially inhomogeneous problem~\cite{delitala2007M3AS,fermo2013SIAP} to track ``risk waves'' along the road; or the problem on networks~\cite{fermo2015M3AS} to study the propagation of the risk on a set of interconnected roads; or even the impact of different types of vehicles, which form a ``traffic mixture''~\cite{puppo2016CMS}, on the risk distribution. On the other hand, the ideas presented in this paper may constitute the basis for modelling risk and safety aspects also of other systems of interacting agents particularly interested by such issues. It is the case of e.g., human crowds, for which a quite wide, though relatively recent, literature already exists (see~\cite[Chapter 4]{cristiani2014BOOK} for a survey), that in some cases~\cite{agnelli2015M3AS} uses a kinetic formalism close to the one which inspired the present work.

\appendix

\section{Basic theory of the kinetic model in Wasserstein spaces}
\label{app:basictheo}
Equation~\eqref{eq:kinetic_model}, complemented with a suitable initial condition, produces the following Cauchy problem:
\begin{equation}
	\begin{cases}
		\partial_t f=Q^{+}(f,\,f)-\rho f, & t>0,\,\w\in I^2 \\[1mm]
		f(0,\,\w)=f_0(\w), & \w\in I^2
	\end{cases}
	\label{eq:cauchy}
\end{equation}
with the compatibility condition $\int_{I^2}f_0(\w)\,d\w=\rho$. Recall that $I^2=[0,\,1]^2\subset\R^2$ is the space of the microscopic states. The problem can be rewritten in mild form by multiplying both sides of the equation by $e^{\rho t}$ and integrating in time:
\begin{align}
	\begin{aligned}[t]
		f(t,\,\w) &= e^{-\rho t}f_0(\w) \\
		&\phantom{=} +\int_0^t e^{\rho(s-t)}\iint_{I^4}\cP(\w_\ast\to\w\,\vert\,\w^\ast,\,\rho)f(s,\,\w_\ast)f(s,\,\w^\ast)\,d\w_\ast\,d\w^\ast\,ds,
	\end{aligned}
	\label{eq:mild}
\end{align}
where we have used that, in view of~\eqref{eq:tabgames}, $\rho=\int_{I^2}f(t,\,\w)\,d\w$ is constant in $t$.

In order to allow for measure-valued kinetic distribution functions, as it happens in the model discussed from Section~\ref{sec:discrete} onwards, an appropriate space in which to study~\eqref{eq:mild} is $X:=C([0,\,T];\,\cM{\rho})$, where $T>0$ is a final time and $\cM{\rho}$ is the space of positive measures on $I^2$ having mass $\rho\geq 0$. An element $f\in X$ is then a continuous mapping $t\mapsto f(t)$, where, for all $t\in [0,\,T]$, $f(t)$ is a positive measure with $\int_{I^2}f(t,\,\w)\,d\w=\rho$.

$X$ is a complete metric space with the distance $\sup\limits_{t\in [0,\,T]}\wass{f(t)}{g(t)}$, where
\begin{equation}
	\wass{f(t)}{g(t)}=\sup_{\varphi\in\Lip_1(I^2)}\int_{I^2}\varphi(\w)(g(t,\,\w)-f(t,\,\w))\,d\w.
	\label{eq:wass}
\end{equation}
is the \emph{1-Wasserstein distance} between $f(t),\,g(t)\in\cM{\rho}$. In particular,
$$ \Lip_1(I^2)=\{\varphi\in C(I^2)\,:\,\Lip(\varphi)\leq 1\}, $$
$\Lip(\varphi)$ denoting the Lipschitz constant of $\varphi$.

\begin{remark}
The definition~\eqref{eq:wass} of $W_1$ follows from the Kantorovich-Rubinstein duality formula, see e.g.,~\cite[Chapter 7]{ambrosio2008BOOK}. However, since in $\cM{\rho}$ all measures carry the same mass and, furthermore, the domain $I^2$ is bounded with $\operatorname{diam}(I^2)\leq 2$, the supremum at the right-hand side is actually the same as that computed over the smaller set $C_{b,1}(I^2)\cap\Lip_1(I^2)$, where
$$ C_{b,1}(I^2)=\{\varphi\in C(I^2)\,:\,\norm{\varphi}_\infty\leq 1\}, $$
see~\cite[Chapter 1]{ulikowska2013PhDTHESIS}. Hence we also have:
\begin{equation}
	\wass{f(t)}{g(t)}=\sup_{\varphi\in C_{b,1}(I^2)\cap\Lip_1(I^2)}\int_{I^2}\varphi(\w)(g(t,\,\w)-f(t,\,\w))\,d\w.
	\label{eq:wass_Cb1}
\end{equation}
In the following we will use both~\eqref{eq:wass} and~\eqref{eq:wass_Cb1} interchangeably.
\end{remark}

\begin{remark} \label{rem:wass_lambda}
Let $\lambda>0$ and let $\varphi\in C(I^2)$ with $\Lip(\varphi)\leq\lambda$. Then
\begin{align*}
	\int_{I^2}\varphi(\w)(g(t,\,\w)-f(t,\,\w))\,d\w &= \lambda\int_{I^2}\frac{\varphi(\w)}{\lambda}(g(t,\,\w)-f(t,\,\w))\,d\w \\
	&\leq \lambda\wass{f(t)}{g(t)},
\end{align*}
considering that $\frac{\varphi}{\lambda}\in\Lip_1(I^2)$.

We will occasionally use this property in the proofs of the forthcoming theorems. Notice that, if $\lambda<1$, this estimate is stricter than that obtained by using directly the fact that $\varphi\in\Lip_1(I^2)$ (i.e., the one without $\lambda$ at the right-hand side).
\end{remark}

To establish the next results, we will always assume that the transition probability distribution $\cP$ satisfies the following Lipschitz continuity property:
\begin{assumption} \label{ass:Lip_P}
Let $\cP(\w_\ast\to\cdot\,\vert\,\w^\ast,\,\rho)\in\sP(I^2)$ for all $\w_\ast,\,\w^\ast\in I^2$, $\rho\in [0,\,1]$, where $\sP(I^2)$ is the space of probability measures on $I^2$. We assume that there exists $\Lip(\cP)>0$, \emph{which may depend on} $\rho$ (although we do not write such a dependence explicitly), such that
$$ \wass{\cP(\w_{\ast 1}\to\cdot\,\vert\,\w^\ast_1,\,\rho)}{\cP(\w_{\ast 2}\to\cdot\,\vert\,\w^\ast_2,\,\rho)}
	\leq\Lip(\cP)\left(\abs{\w_{\ast 2}-\w_{\ast 1}}+\abs{\w^\ast_2-\w^\ast_1}\right) $$
for all $\w_{\ast 1},\,\w_{\ast 2},\,\w^\ast_1,\,\w^\ast_2\in I^2$ and all $\rho\in [0,\,1]$.
\end{assumption}

\subsection{Existence and uniqueness of the solution}
Taking advantage of the mild formulation~\eqref{eq:mild} of the problem, we apply Banach fixed-point theorem in $X$ to prove:

\begin{theorem}	\label{theo:exist_uniqueness}
Fix $\rho\in [0,\,1]$ and let $f_0\in\cM{\rho}$. There exists a unique $f\in C([0,\,+\infty);\,\cM{\rho})$ which solves~\eqref{eq:mild}.
\end{theorem}
\begin{proof}
We assume $\rho>0$ for, if $\rho=0$, the unique solution to~\eqref{eq:mild} is clearly $f\equiv 0$ and we are done. We fix $T>0$ and we introduce the operator $\cS$ defined on $X$ as
\begin{align*}
	\cS(f)(t,\,\w) &:= e^{-\rho t}f_0(\w) \\
	&\phantom{:=} +\int_0^t e^{\rho(s-t)}\iint_{I^4}\cP(\w_\ast\to\w\,\vert\,\w^\ast,\,\rho)f(s,\,\w_\ast)f(s,\,\w^\ast)\,d\w_\ast\,d\w^\ast\,ds.
\end{align*}
Then we restate~\eqref{eq:mild} as $f=\cS(f)$, meaning that solutions to~\eqref{eq:mild} are fixed points of $\cS$ on $X$. Now we claim that:
\begin{itemize}
\item $\cS(X)\subseteq X$. \\
Let $f\in X$. The non-negativity of $f_0$ (by assumption) and that of $\cP$ (by construction) give immediately that $\cS(f)(t)$ is a positive measure for all $t\in [0,\,T]$. Moreover, a simple calculation using property~\eqref{eq:tabgames} shows that the mass of $\cS(f)(t)$ is
$$ \int_{I^2}\cS(f)(t,\,\w)\,d\w=\rho e^{-\rho t}+\rho^2\int_0^t e^{\rho(s-t)}\,ds=\rho. $$
Therefore we conclude that $\cS(f)(t)\in \cM{\rho}$ for all $t\in [0,\,T]$.

To check the continuity of the mapping $t\mapsto\cS(f)(t)$ we define
$$ \cI(\varphi)(s):=\iint_{I^4}\left(\int_{I^2}\varphi(\w)\cP(\w_\ast\to\w\,\vert\,\w^\ast,\,\rho)\,d\w\right)
	f(s,\,\w_\ast)f(s,\,\w^\ast)\,d\w_\ast\,d\w^\ast, $$
for $\varphi\in C_{b,1}(I^2)\cap\Lip_1(I^2)$, then we take $t_1,\,t_2\in [0,\,T]$ with, say, $t_1\leq t_2$ and we compute:
\begin{align*}
	\wass{\cS(f)(t_1)}{\cS(f)(t_2)} &= \sup_{\varphi\in C_{b,1}(I^2)\cap\Lip_1(I^2)}\Biggl[\left(e^{-\rho t_2}-e^{-\rho t_1}\right)\int_{I^2}\varphi(\w)f_0(\w)\,d\w \\
	&\phantom{=} +\int_0^{t_2}e^{\rho(s-t_2)}\cI[\varphi](s)\,ds-\int_0^{t_1}e^{\rho(s-t_1)}\cI[\varphi](s)\,ds\Biggr] \\
	&\leq \rho\abs{e^{-\rho t_2}-e^{-\rho t_1}}\left(1+\rho\int_0^{t_1}e^{\rho s}\,ds\right)+\rho^2e^{-\rho t_2}\int_{t_1}^{t_2}e^{\rho s}\,ds \\
	&\leq 3\rho^2\abs{t_2-t_1},
\end{align*}
where we have used the Lipschitz continuity of the exponential function and the fact that $\abs{\cI(\varphi)(s)}\leq\rho^2$. Finally, this says that $\cS(f)\in X$.
\item If $T>0$ is sufficiently small then $\cS$ is a contraction on $X$. \\
Let $f,\,g\in X$, $\varphi\in\Lip_1(I^2)$, and define
\begin{equation}
	\psi(\w_\ast,\,\w^\ast):=\int_{I^2}\varphi(\w)\cP(\w_\ast\to\w\,\vert\,\w^\ast,\,\rho)\,d\w.
	\label{eq:psi}
\end{equation}
Notice preliminarily that, owing to Assumption~\ref{ass:Lip_P},
$$ \abs{\psi(\w_{\ast 2},\,\w^\ast_2)-\psi(\w_{\ast 1},\,\w^\ast_1)}\leq\Lip(\cP)\left(\abs{\w_{\ast 2}-\w_{\ast 1}}+\abs{\w^\ast_2-\w^\ast_1}\right). $$
Then:
\begin{align*}
	\int_{I^2} & \varphi(\w)(\cS(g)(t,\,\w)-\cS(f)(t,\,\w))\,d\w \\
	&= \int_0^t e^{\rho(s-t)}\iint_{I^4}\psi(\w_\ast,\,\w^\ast)(g_\ast g^\ast-f_\ast f^\ast)\,d\w_\ast\,d\w^\ast\,ds
\intertext{($f_\ast$, $f^\ast$ being shorthand for $f(s,\,\w_\ast)$, $f(s,\,\w^\ast)$ and analogously $g_\ast$, $g^\ast$)}
	&= \int_0^t e^{\rho(s-t)}\int_{I^2}\left(\int_{I^2}\psi(\w_\ast,\,\w^\ast)g_\ast\,d\w_\ast\right)(g^\ast-f^\ast)\,d\w^\ast\,ds \\
	&\phantom{=} +\int_0^t e^{\rho(s-t)}\int_{I^2}\left(\int_{I^2}\psi(\w_\ast,\,\w^\ast)f^\ast\,d\w^\ast\right)(g_\ast-f_\ast)\,d\w_\ast\,ds.
\end{align*}
Notice that
$$ \abs{\int_{I^2}\Bigl(\psi(\w_\ast,\,\w^\ast_2)-\psi(\w_\ast,\,\w^\ast_1)\Bigr)g_\ast\,d\w_\ast}\leq
	\rho\Lip(\cP)\abs{\w^\ast_2-\w^\ast_1} $$
and that the same holds also for the mapping $\w_\ast\mapsto\int_{I^2}\psi(\w_\ast,\,\w^\ast)f^\ast\,d\w^\ast$. Hence we continue the previous calculation by appealing to Remark~\ref{rem:wass_lambda} (at the right-hand side) and to the arbitrariness of $\varphi$ to discover:
\begin{align*}
	\wass{\cS(f)(t)}{\cS(g)(t)} &\leq 2\rho\Lip{(\cP)}\int_0^t e^{\rho(s-t)}\wass{f(s)}{g(s)}\,ds \\
	&\leq 2\Lip{(\cP)}\left(1-e^{-\rho t}\right)\sup_{t\in [0,\,T]}\wass{f(t)}{g(t)},
\end{align*}
whence finally
$$ \sup_{t\in [0,\,T]}\wass{\cS(f)(t)}{\cS(g)(t)}\leq 2\Lip{(\cP)}\left(1-e^{-\rho T}\right)\sup_{t\in [0,\,T]}\wass{f(t)}{g(t)}. $$
From this inequality we see that:
\begin{enumerate}[$-$]
\item if $\Lip{(\cP)}>\frac{1}{2}$ then it suffices to take $T<\frac{1}{\rho}\log{\frac{2\Lip{(\cP)}}{2\Lip{(\cP)}-1}}$ to obtain that $\cS$ is a contraction on $X$;
\item if $\Lip{(\cP)}\leq\frac{1}{2}$ then $\cS$ is a contraction on $X$ for every $T>0$.
\end{enumerate} 
\end{itemize}

Owing to the properties above, Banach fixed-point theorem implies the existence of a unique fixed point $f\in X$ of $\cS$ which solves~\eqref{eq:mild}. If $\Lip{(\cP)}\leq\frac{1}{2}$ this solution is global in time, whereas if $\Lip{(\cP)}>\frac{1}{2}$ it is only local. However, a simple continuation argument, based on taking $f(T)\in\cM{\rho}$ as new initial condition for $t=T$ and repeating the procedure above, shows that we can extend it uniquely on the interval $[T,\,2T]$. Proceeding in this way, we do the same on all subsequent intervals of the form $[kT,\,(k+1)T]$, $k=2,\,3,\,\dots$, and we obtain also in this case a global-in-time solution.
\end{proof}

\subsection{Continuous dependence}
By comparing two solutions to~\eqref{eq:mild} carrying the same mass $\rho$ we can establish:

\begin{theorem} 	\label{theo:cont_dep}
Fix $\rho\in [0,\,1]$ and two initial data $f_{01},\,f_{02}\in\cM{\rho}$. Let $f_1,\,f_2\in C([0,\,+\infty);\,\cM{\rho})$ be the corresponding solution to~\eqref{eq:mild}. Then:
$$ \wass{f_1(t)}{f_2(t)}\leq e^{-\rho(1-2\Lip(\cP))t}\wass{f_{01}}{f_{02}} \quad \forall\,t\geq 0. $$
\end{theorem}
\begin{proof}\footnote{Throughout the proof, we will adopt the shorthand notations $f_\ast=f(t,\,\w_\ast)$, $f^\ast=f(t,\,\w^\ast)$.}
Let $\varphi\in\Lip_1(I^2)$. Using~\eqref{eq:mild} we compute:
\begin{multline*}
	\int_{I^2}\varphi(\w)(f_2(t,\,\w)-f_1(t,\,\w))\,d\w\leq e^{-\rho t}\wass{f_{01}}{f_{02}} \\
	+\int_0^t e^{\rho(s-t)}\iint_{I^4}\psi(\w_\ast,\,\w^\ast)\left(f_{2\ast}f_2^\ast-f_{1\ast}f_1^\ast\right)\,d\w_\ast\,d\w^\ast,
\end{multline*}
where $\psi(\w_\ast,\,\w^\ast)$ is the function~\eqref{eq:psi} defined in the proof of Theorem~\ref{theo:exist_uniqueness}. By means of analogous calculations we discover
\begin{align*}
	\wass{f_1(t)}{f_2(t)} &\leq e^{-\rho t}\wass{f_{01}}{f_{02}} \\
	&\phantom{\leq} +2\rho\Lip(\cP)\int_0^t e^{\rho(s-t)}\wass{f_1(s)}{f_2(s)}\,ds,
\end{align*}
whence the thesis follows by applying Gronwall's inequality.
\end{proof}

\begin{remark}
From Theorem~\ref{theo:cont_dep} we infer that if $\Lip(\cP)<\frac{1}{2}$ then
$$ \lim_{t\to+\infty}\wass{f_1(t)}{f_2(t)}=0. $$
That is, all solutions to~\eqref{eq:mild} approach asymptotically one another. This fact preludes to the result about the equilibria of~\eqref{eq:mild} proved in Theorem~\ref{theo:attractiveness} below.
\end{remark}

\subsection{Asymptotic analysis}
\label{app:equilibria}
In this section we study the asymptotic trends of~\eqref{eq:kinetic_model}, in particular we give sufficient conditions for the existence, uniqueness, and attractiveness of equilibria. It is worth stressing that equilibria of the kinetic model are at the basis of the computation of fundamental and risk diagrams of traffic discussed in Section~\ref{sec:simulations}.

Besides the methods presented here, we refer the reader to~\cite{herty2010KRM} and references therein for other ways to study the trend towards equilibrium of space homogeneous kinetic traffic models and for the identification of exact or approximated steady states.

\subsubsection{Existence and uniqueness of equilibria}
Equilibria of~\eqref{eq:kinetic_model} are time-independent distribution functions $f_\infty=f_\infty(\w)\in\cM{\rho}$ such that
$$ Q^+(f_\infty,\,f_\infty)-\rho f_\infty=0. $$
If $\rho=0$ then it is clear that the unique equilibrium is the trivial distribution function $f_\infty\equiv 0$. Assuming instead $\rho>0$, from the previous equation we see that equilibria satisfy
\begin{equation}
	f_\infty=\frac{1}{\rho}Q^+(f_\infty,\,f_\infty),
	\label{eq:equilibria}
\end{equation}
i.e., they are fixed points of the mapping $f\mapsto\frac{1}{\rho}Q^+(f,\,f)$. The next theorem gives a sufficient condition for their existence and uniqueness, relying on the Banach contraction principle in $\cM{\rho}$.

\begin{theorem} \label{theo:equilibria}
Let $\Lip(\cP)<\frac{1}{2}$. For all $\rho\in [0,\,1]$,~\eqref{eq:kinetic_model} admits a unique equilibrium distribution function $f_\infty\in\cM{\rho}$.
\end{theorem}
\begin{proof}
Throughout the proof we will assume $\rho>0$.

The operator $\frac{1}{\rho}Q^+$ maps $\cM{\rho}$ into itself, in fact, given $f\in\cM{\rho}$, it is clear that $\frac{1}{\rho}Q^+(f,\,f)$ is a positive measure and moreover
$$ \int_{I^2}\frac{1}{\rho}Q^+(f,\,f)(\w)\,d\w=\frac{1}{\rho}\iint_{I^4}f(\w_\ast)f(\w^\ast)\,d\w_\ast\,d\w^\ast=\rho. $$

Moreover we claim that, under the assumptions of the theorem, it is a contraction on $\cM{\rho}$. Indeed, let $f,\,g\in\cM{\rho}$ and $\varphi\in\Lip_1(I^2)$, then:
\begin{align*}
	\frac{1}{\rho}\int_{I^2}\varphi(\w) & \left(Q^+(g,\,g)(\w)-Q^+(f,\,f)(\w)\right)\,d\w \\
	&= \frac{1}{\rho}\int_{I^2}\left(\int_{I^2}\psi(\w_\ast,\,\w^\ast)g(\w_\ast)\,d\w_\ast\right)(g(\w^\ast)-f(\w^\ast))\,d\w^\ast \\
	&\phantom{=} +\frac{1}{\rho}\int_{I^2}\left(\int_{I^2}\psi(\w_\ast,\,\w^\ast)f(\w^\ast)\,d\w^\ast\right)(g(\w_\ast)-f(\w_\ast))\,d\w_\ast,
\end{align*}
where $\psi$ is the function~\eqref{eq:psi}. From the proof of Theorem~\ref{theo:exist_uniqueness}, we know that both mappings $\w^\ast\mapsto\int_{I^2}\psi(\w_\ast,\,\w^\ast)g(\w_\ast)\,d\w_\ast$ and $\w_\ast\mapsto\int_{I^2}\psi(\w_\ast,\,\w^\ast)f(\w^\ast)\,d\w^\ast$ are Lipschitz continuous with Lipschitz constant bounded by $\rho\Lip(\cP)$, hence from the previous expression we deduce
$$ \frac{1}{\rho}\int_{I^2}\varphi(\w)\left(Q^+(g,\,g)(\w)-Q^+(f,\,f)(\w)\right)\,d\w\leq 2\Lip(\cP)\wass{f}{g}. $$
Taking the supremum over $\varphi$ at the left-hand side yields finally
$$ \wass{\frac{1}{\rho}Q^+(f,\,f)}{\frac{1}{\rho}Q^+(g,\,g)}\leq 2\Lip(\cP)\wass{f}{g}, $$
which, in view of the hypothesis $\Lip(\cP)<\frac{1}{2}$, implies that $\frac{1}{\rho}Q^+$ is a contraction. Banach fixed-point theorem gives then the thesis.
\end{proof}

\subsubsection{Attractiveness of equilibria}
Under the same assumption of Theorem~\ref{theo:equilibria}, the equilibrium distribution function $f_\infty$ is globally attractive. This means that all solutions to~\eqref{eq:cauchy} converge to $f_\infty$ asymptotically in time. The precise statement of the result is as follows:

\begin{theorem} \label{theo:attractiveness}
Let $\Lip(\cP)<\frac{1}{2}$. Any solution $f\in C([0,\,+\infty);\,\cM{\rho})$ of~\eqref{eq:cauchy} converges to $f_\infty$ in the $W_1$ metric when $t\to+\infty$.
\end{theorem}
\begin{proof}
Set $f_0(\w)=f(0,\,\w)\in\cM{\rho}$. Since $f_\infty$ is the solution to~\eqref{eq:cauchy} when the initial datum is $f_\infty$ itself, Theorem~\ref{theo:cont_dep} implies
$$ \wass{f(t)}{f_\infty}\leq e^{-\rho(1-2\Lip(\cP))t}\wass{f_0}{f_\infty} $$
and the thesis follows.
\end{proof}

\bibliographystyle{amsplain}
\bibliography{FpTa-safety_traffic}

\providecommand{\bysame}{\leavevmode\hbox to3em{\hrulefill}\thinspace}
\providecommand{\MR}{\relax\ifhmode\unskip\space\fi MR }
% \MRhref is called by the amsart/book/proc definition of \MR.
\providecommand{\MRhref}[2]{%
  \href{http://www.ams.org/mathscinet-getitem?mr=#1}{#2}
}
\providecommand{\href}[2]{#2}
\begin{thebibliography}{10}

\bibitem{DaCoTA2012}
\emph{Annual {S}tatistical {R}eport 2012}, Tech. report, {E}uropean {R}oad
  {S}afety {O}bservatory, 2012.

\bibitem{KiwiRAP2012}
\emph{How safe are our roads? -- {T}racking the safety performance of {N}ew
  {Z}ealand's state highway network}, Tech. report, KiwiRAP, 2012.

\bibitem{abdel-aty2000AAP}
M.~A. Abdel-Aty and A.~E. Radwan, \emph{Modeling traffic accident occurrence
  and involvement}, Accident Anal. Prev. \textbf{32} (2000), no.~5, 633--642.

\bibitem{agnelli2015M3AS}
J.~P. Agnelli, F.~Colasuonno, and D.~Knopoff, \emph{A kinetic theory approach
  to the dynamics of crowd evacuation from bounded domains}, Math. Models
  Methods Appl. Sci. \textbf{25} (2015), no.~1, 109--129.

\bibitem{ambrosio2008BOOK}
L.~Ambrosio, N.~Gigli, and G.~Savar{\'e}, \emph{Gradient flows in metric spaces
  and in the space of probability measures}, Lectures in Mathematics ETH
  Z\"urich, Birkh\"auser Verlag, Basel, 2008.

\bibitem{bellomo2015NHM}
N.~Bellomo, F.~Colasuonno, D.~Knopoff, and J.~Soler, \emph{From a systems
  theory of sociology to modeling the onset and evolution of criminality},
  Netw. Heterog. Media \textbf{10} (2015), no.~3, 421--441.

\bibitem{bonzani2003MCM}
I.~Bonzani and L.~Mussone, \emph{From experiments to hydrodynamic traffic flow
  models. {I}. {M}odelling and parameter identification}, Math. Comput.
  Modelling \textbf{37} (2003), no.~12-13, 1435--1442.

\bibitem{colasuonno2013CM}
F.~Colasuonno and M.~C. Salvatori, \emph{Existence and uniqueness of solutions
  to a {C}auchy problem modeling the dynamics of socio-political conflicts},
  Contemp. Math. \textbf{594} (2013), 155--165.

\bibitem{cristiani2014BOOK}
E.~Cristiani, B.~Piccoli, and A.~Tosin, \emph{Multiscale {M}odeling of
  {P}edestrian {D}ynamics}, MS\&A: Modeling, Simulation and Applications,
  vol.~12, Springer International Publishing, 2014.

\bibitem{delitala2007M3AS}
M.~Delitala and A.~Tosin, \emph{Mathematical modeling of vehicular traffic: a
  discrete kinetic theory approach}, Math. Models Methods Appl. Sci.
  \textbf{17} (2007), no.~6, 901--932.

\bibitem{evans1986RA}
L.~Evans, \emph{Risk homeostasis theory and traffic accident data}, Risk Anal.
  \textbf{6} (1986), no.~1, 81--94.

\bibitem{fermo2013SIAP}
L.~Fermo and A.~Tosin, \emph{A fully-discrete-state kinetic theory approach to
  modeling vehicular traffic}, SIAM J. Appl. Math. \textbf{73} (2013), no.~4,
  1533--1556.

\bibitem{fermo2014DCDSS}
\bysame, \emph{Fundamental diagrams for kinetic equations of traffic flow},
  Discrete Contin. Dyn. Syst. Ser. S \textbf{7} (2014), no.~3, 449--462.

\bibitem{fermo2015M3AS}
\bysame, \emph{A fully-discrete-state kinetic theory approach to traffic flow
  on road networks}, Math. Models Methods Appl. Sci. \textbf{25} (2015), no.~3,
  423--461.

\bibitem{hermans2009AAP}
E.~Hermans, T.~Brijs, G.~Wets, and K.~Vanhoof, \emph{Benchmarking road safety:
  {L}essons to learn from a data envelopment analysis}, Accident Anal. Prev.
  \textbf{41} (2009), no.~1, 174--182.

\bibitem{hermans2008AAP}
E.~Hermans, F.~Van~den Bossche, and G.~Wets, \emph{Combining road safety
  information in a performance index}, Accident Anal. Prev. \textbf{40} (2008),
  no.~4, 1337--1344.

\bibitem{herty2010KRM}
M.~Herty and L.~Pareschi, \emph{{F}okker-{P}lanck asymptotics for traffic flow
  models}, Kinet. Relat. Models \textbf{3} (2010), no.~1, 165--179.

\bibitem{herty2011ZAMM}
M.~Herty and V.~Schleper, \emph{Traffic flow with unobservant drivers}, Z.
  Angew. Math. Mech. \textbf{91} (2011), no.~10, 763--776.

\bibitem{kerner2004BOOK}
B.~S. Kerner, \emph{The physics of traffic}, Understanding Complex Systems,
  Springer, Berlin, 2004.

\bibitem{li2011TRR}
J.~Li and M.~Zhang, \emph{Fundamental diagram of traffic flow}, Transp. Res.
  Record \textbf{2260} (2011), 50--59.

\bibitem{miaou1993AAP}
S.-P. Miaou and H.~Lum, \emph{Modeling vehicle accidents and highway geometric
  design relationships}, Accident Anal. Prev. \textbf{25} (1993), no.~6,
  689--709.

\bibitem{moutari2014CMS}
S.~Moutari and M.~Herty, \emph{A {L}agrangian approach for modeling road
  collisions using second-order models of traffic flow}, Commun. Math. Sci.
  \textbf{12} (2014), no.~7, 1239--1256.

\bibitem{moutari2013IMAJAM}
S.~Moutari, M.~Herty, A.~Klein, M.~Oeser, B.~Steinauer, and V.~Schleper,
  \emph{Modelling road traffic accidents using macroscopic second-order models
  of traffic flow}, IMA J. Appl. Math. \textbf{78} (2013), no.~5, 1087--1108.

\bibitem{oppe1989AAP}
S.~Oppe, \emph{Macroscopic models for traffic and traffic safety}, Accident
  Anal. Prev. \textbf{21} (1989), no.~3, 225--232.

\bibitem{piccoli2009ENCYCLOPEDIA}
B.~Piccoli and A.~Tosin, \emph{Vehicular traffic: {A} review of continuum
  mathematical models}, Encyclopedia of Complexity and Systems Science (R.~A.
  Meyers, ed.), vol.~22, Springer, New York, 2009, pp.~9727--9749.

\bibitem{pucci2014DCDSS}
P.~Pucci and M.~C. Salvatori, \emph{On an initial value problem modeling
  evolution and selection in living systems}, Discrete Contin. Dyn. Syst. Ser.
  S \textbf{7} (2014), no.~4, 807--821.

\bibitem{puppo2016CMS}
G.~Puppo, M.~Semplice, A.~Tosin, and G.~Visconti, \emph{Fundamental diagrams in
  traffic flow: the case of heterogeneous kinetic models}, Commun. Math. Sci.
  \textbf{14} (2016), no.~3, 643--669.

\bibitem{shen2012AAP}
Y.~Shen, E.~Hermans, T.~Brijs, G.~Wets, and K.~Vanhoof, \emph{Road safety risk
  evaluation and target setting using data envelopment analysis and its
  extensions}, Accident Anal. Prev. \textbf{48} (2012), 430--441.

\bibitem{tosin2009AML}
A.~Tosin, \emph{From generalized kinetic theory to discrete velocity modeling
  of vehicular traffic. {A} stochastic game approach}, Appl. Math. Lett.
  \textbf{22} (2009), no.~7, 1122--1125.

\bibitem{ulikowska2013PhDTHESIS}
A.~Ulikowska, \emph{Structured population models in metric spaces}, Ph.D.
  thesis, Warsaw University, Warsaw, {P}oland, May 2013.

\bibitem{wilde1998IP}
G.~J.~S. Wilde, \emph{Risk homeostasis theory: an overview}, Inj. Prev.
  \textbf{4} (1998), 89--91.

\bibitem{yannis2011TRB}
G.~Yannis, C.~Antoniou, and Papadimitriou E., \emph{Modeling traffic fatalities
  in {E}urope}, Proceedings of the {TRB} 90th {A}nnual {M}eeting (Washington,
  D.C.), January 2011.

\bibitem{zhang2005TRB}
H.~M. Zhang and T.~Kim, \emph{A car-following theory for multiphase vehicular
  traffic flow}, Transport. Res. B-Meth. \textbf{39} (2005), no.~5, 385--399.

\end{thebibliography}

\end{document}